\documentclass[10pt]{article}

\usepackage{latexsym}
\usepackage{amsfonts}
\usepackage{amsmath}
\usepackage{amssymb}
\usepackage{url}
\usepackage{graphicx}
\usepackage{setspace}

\newcommand{\ignore}[1]{}

\newtheorem{remark}{Remark}
\newenvironment{proof}{\noindent{\em Proof}.}{}
\newtheorem{definition}{\mbox{Definition}}[section]
\newtheorem{theorem}{\mbox{Theorem}}[section]
\newtheorem{corollary}{\mbox{Corollary}}[section]
\newtheorem{proposition}{\mbox{Proposition}}[section]

\newcommand{\qed}{{\hfill $\Box$}}

\newcommand{\A}{{\cal A}}

\newcommand{\Ora}{{\cal O}}

\newcommand{\TRUE}{\text {\sc true}}

\newcommand{\FALSE}{\text {\sc false}}

\newcommand{\GC}{\ensuremath{\mathcal{GC}}}
\newcommand{\OO}{\ensuremath{\mathcal{O}}}
\newcommand{\N}{\ensuremath{\mathbb{N}}}
\newcommand{\ID}{\ensuremath{\mathbb{I}\mathbb{D}}}

\newcommand{\SGC}{{\sf SGC}}
\newcommand{\SSGC}{{\sf SSGC}}

\newcommand{\setup}{{\sf Setup}}
\newcommand{\join}{{\sf Join}}
\newcommand{\leave}{{\sf Leave}}
\newcommand{\rekey}{{\sf Rekey}}

\newcommand{\keyset}{\text{\sc keyset}}
\newcommand{\userset}{\text{\sc userset}}

\newcommand{\bk}{\ensuremath{\bar{\sf k}}}
\newcommand{\hk}{\ensuremath{\hat{\sf k}}}

\newcommand{\acc}{{\sf acc}}

\begin{document}

\title{On the Security of Group Communication Schemes\thanks{An
extended abstract of this paper appeared as \cite{XuSASN05}.}}

\author{Shouhuai Xu \\
Department of Computer Science, University of Texas at San Antonio \\
{\tt shxu@cs.utsa.edu}
}

\date{}

\maketitle

\begin{abstract}
Secure group communications are a mechanism facilitating protected
transmission of messages from a sender to multiple receivers, and
many emerging applications in both wired and wireless networks need
the support of such a mechanism. There have been many secure group
communication schemes in wired networks, which can be
directly adopted in, or appropriately adapted to, 
wireless networks such as mobile ad hoc networks (MANETs) and sensor
networks. In this paper we show that the popular group communication
schemes that we have examined are vulnerable to the following
attack: An outside adversary who compromises a certain legitimate
group member could obtain {\em all} past and present group keys (and
thus all the messages protected by them); this is in sharp contrast
to the widely-accepted belief that a such adversary can only obtain
the present group key (and thus the messages protected by it).
In order to understand and deal with the attack, we formalize two security models
for stateful and stateless group communication schemes.
We show that some practical methods can make a {\em
subclass} of existing group communication schemes immune to the attack.
\end{abstract}

\noindent{\bf Keywords}: security, key management, group communication, multicast.

\section{Introduction}

Secure group communications are useful in both wired and wireless
networks, because they facilitate protected transmission of messages
from a sender to multiple receivers. One important property of
secure group communications is to ensure that only the legitimate
members (or users, receivers) can have access to the multicast or
broadcast data. There have been many secure group communication
schemes in the setting of wired networks; popular ones include the
stateful LKH \cite{WGL00,WallnerIETF98} and OFT
\cite{ShermanOFT03,BalensonOFTDraft} as well as stateless ones
\cite{NaorCrypto01,JhoEurocrypt05}. These schemes can be directly
adopted in, or appropriately adapted to, the setting of wireless
networks such as mobile ad hoc networks (MANETs) and sensor
networks. The core component of a secure group communication scheme
is its key management method. A common feature among these schemes'
key management methods is that each user holds a set of keys that
are then utilized to help establish some group keys (which are
common to all the group members and are used to encrypt actual
messages).

In this paper we show that these group communication schemes, or more specifically
their key management methods, are subject to
the following attack: An outside adversary who compromises a
certain legitimate group member could obtain {\em all} past
and present group keys (and thus the data encrypted using these keys).
This is in sharp contrast to the widely-accepted belief
that such an adversary can only obtain the present group key.
This attack is powerful also because it provides the adversary
the following flexibility: There are potentially many
legitimate group members such that compromising any (or a small number)
of them leads to the exposure of both past and present
group keys. This flexibility may be particularly relevant in
the setting of MANETs and sensor networks because they are typically deployed in
a small area and the adversary can capture and compromise the easiest-to-obtain node(s).

\subsection{Motivating Problems}
\label{sec:motivation}

Now we explore some attack scenarios
against the stateful LKH \cite{WGL00,WallnerIETF98} and
OFT \cite{ShermanOFT03,BalensonOFTDraft}, and
against stateless ones \cite{NaorCrypto01,JhoEurocrypt05}.
The emphasis is on the case of LKH.

\smallskip

\noindent{\bf Vulnerability of the LKH and LKH+ schemes:}
Let's first briefly review the LKH group communication scheme.
Following the notations of \cite{WGL00}, we let
$$
x \to \{y_1, \ldots, y_\ell \} : \{z\}_w
$$
denote that $x$ sends the users $y_1, \ldots, y_\ell$ (via multicast or unicast)
the encryption of plaintext $z$ using key $w$, namely the ciphertext $\{z\}_w$.

Consider the simple scenario, as shown in Figure
\ref{fig:attack}.(a), of a group consisting of a key server $s$ and
users $u_1, \ldots, u_8$. The server is responsible for initiating
and maintaining the group in the presence of user dynamics (i.e.,
joins and leaves). The keys are organized as a {\em key tree}, where
the leaves are the users and the inner nodes are the keys. Moreover,
each user holds the keys corresponding to the inner nodes on the
path starting from the parent of the user and ending at the root.
For example, in Figure \ref{fig:attack}.(a), user $u_1$ holds keys
$k_1$, $k_{123}$, and $k_{1-8}$, where $k_{1-8}$ is the {\em group
key} that can be used to encrypt the communications within the
group.

\begin{figure*}[ht]
\begin{center}
\includegraphics[height=6.0in]{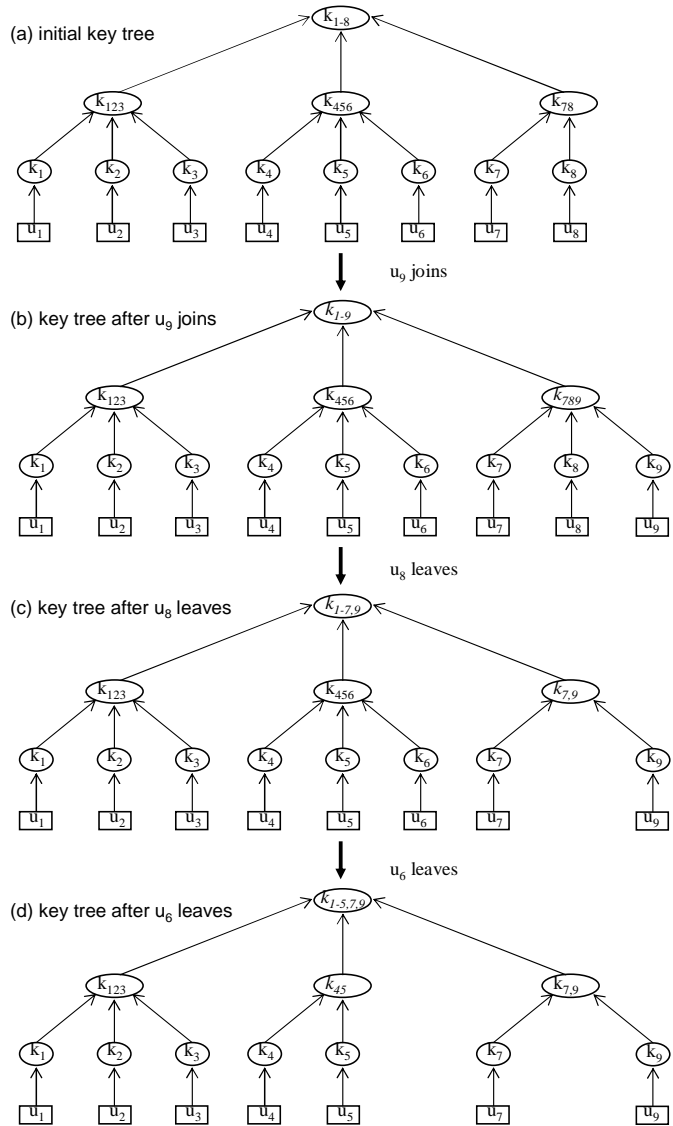}
\end{center}
\caption{A scenario of LKH}
\label{fig:attack}
\end{figure*}

In order to maintain secure communications, each join or leave would require the key server to
change some keys that also need to be securely distributed to certain users (via some
{\em rekeying messages}).
Ignoring for a moment certain details such as authorization of joining the group and
authentication of the messages sent by the key server, in what follows we explain how the key
server responds to group dynamics.

After granting a join request from user $u_9$, server $s$ shares a key $k_9$ with user $u_9$.
Certain keys need to be changed and sent to certain relevant users.
As shown in Figure \ref{fig:attack}.(b), in order to
prevent $u_9$ from accessing past communications,
$k_{78}$ and $k_{1-8}$ are changed to $k_{789}$ and $k_{1-9}$, respectively.
Moreover, the new group key $k_{1-9}$ needs to be securely sent to users $u_1,\ldots,u_9$, and
$k_{789}$ needs to be securely sent to users $u_7$, $u_8$, and $u_9$.
One efficient way to do this is the following algorithm (which corresponds to the so-called
{\em group-oriented rekeying} strategy):
\begin{eqnarray*}
\begin{array}{lll}
s \to \{u_1,\ldots,u_8\} & : & \{k_{1-9}\}_{k_{1-8}}, \{k_{789}\}_{k_{78}} \\
s \to \{u_9\} & : & \{k_{1-9},k_{789}\}_{k_9}
\end{array}
\end{eqnarray*}
Furthermore, $k_{1-8}$ is securely erased by $u_1,\ldots,u_8$, and
$k_{78}$ is securely erased by $u_7$ and $u_8$.

Now suppose $u_8$ leaves. To prevent $u_8$ from accessing future communications,
as shown in Figure \ref{fig:attack}.(c), server $s$ needs to change
$k_{1-9}$ and $k_{789}$ to $k_{1-7,9}$ and $k_{7,9}$, respectively.
Moreover, the new group key
$k_{1-7,9}$ needs to be securely sent to users $u_1, \ldots,u_7,u_9$,
and $k_{7,9}$ needs to be securely sent to $u_7$ and $u_9$.
One efficient way to do this is the following algorithm (which also corresponds to
the group-oriented rekeying strategy):
\begin{eqnarray*}
\begin{array}{lll}
s \to \{u_1,\ldots,u_7,u_9\} & : & \{k_{1-7,9}\}_{k_{123}},
\{k_{1-7,9}\}_{k_{456}}, \{k_{1-7,9}\}_{k_{7,9}},
\{k_{7,9}\}_{k_7}, \{k_{7,9}\}_{k_9}
\end{array}
\end{eqnarray*}
Furthermore, $k_{1-9}$ is securely erased by $u_1,\ldots,u_7,u_9$, and
$k_{789}$ is securely erased by $u_7$ and $u_9$.

Now suppose $u_6$ leaves also. To prevent $u_6$ from accessing future communications,
as shown in Figure \ref{fig:attack}.(d), server $s$ needs to change
$k_{1-7,9}$ and $k_{456}$ to $k_{1-5,7,9}$ and $k_{45}$, respectively.
Moreover, the new group key
$k_{1-5,7,9}$ needs to be securely sent to users $u_1, \ldots, u_5, u_7, u_9$,
and $k_{45}$ needs to be securely sent to users $u_4$ and $u_5$.
One efficient way to do this is the following algorithm (which also corresponds to
the group-oriented rekeying strategy):
\begin{eqnarray*}
\begin{array}{lll}
s \to \{u_1,\ldots,u_5,u_7,u_9\} & : & \{k_{1-5,7,9}\}_{k_{123}},
\{k_{1-5,7,9}\}_{k_{45}},  \{k_{1-5,7,9}\}_{k_{7,9}},
\{k_{45}\}_{k_{4}}, \{k_{45}\}_{k_5}
\end{array}
\end{eqnarray*}
Furthermore, $k_{1-7,9}$ is securely erased by $u_1,\ldots,u_5,u_7,u_9$, and
$k_{456}$ is securely erased by $u_4$ and $u_5$.

Given the above system setting,
let us now examine the consequences of a legitimate user being compromised.
\begin{itemize}
\item Suppose an adversary compromises user $u_9$.
It is of course true that the adversary is able to
obtain the present group key $k_{1-5,7,9}$,
no matter how the group rekeying scheme works.
Moreover, the adversary can obtain $k_{7,9}$ and $k_9$.
We observe that the adversary who has recorded the network traffic
is also able to obtain the past group keys $k_{1-9}$ and $k_{1-7,9}$,
because it can decrypt the messages incurred by the events
that $u_9$ joins the group and that $u_8$ leaves the group:
\begin{eqnarray*}
\begin{array}{lll}
s \to \{u_9\} & : & \{k_{1-9},k_{789}\}_{k_9}, \\
s \to \{u_1,\ldots,u_7,u_9\} & : & \{k_{1-7,9}\}_{k_{123}},
\{k_{1-7,9}\}_{k_{456}}, \{k_{1-7,9}\}_{k_{7,9}},
\{k_{7,9}\}_{k_7}, \{k_{7,9}\}_{k_9}.
\end{array}
\end{eqnarray*}
We stress that this is true even though the past group keys $k_{1,9}$ and
$k_{1-7,9}$ were securely erased by $u_9$.
As a consequence, the adversary can decrypt the communications encrypted using
the past and present group keys
$k_{1-9}$, $k_{1-7,9}$, and $k_{1-5,7,9}$. We notice that
the initial group key $k_{1-8}$ is never accessible to $u_9$.

\item Suppose an adversary compromises user $u_7$. Then, the adversary knows
$k_{1-5,7,9}$, $k_{7,9}$, and $k_7$. Note that the adversary can
obtain $k_{1-7,9}$ from the recorded traffic corresponding to the event that
$u_8$ leaves the group:
\begin{eqnarray*}
\begin{array}{lll}
s \to \{u_1,\ldots,u_7,u_9\} & : & \{k_{1-7,9}\}_{k_{123}},
\{k_{1-7,9}\}_{k_{456}}, \{k_{1-7,9}\}_{k_{7,9}},
\{k_{7,9}\}_{k_7}, \{k_{7,9}\}_{k_9}.
\end{array}
\end{eqnarray*}
We stress that this is true even though the past group key
$k_{1-7,9}$ was securely erased by $u_7$.
We notice that the above analysis is based on the implicit assumption that the initial
group key $k_{1-8}$ was ``magically" sent to $u_7$. In practice, $k_{1-8}$
might have been sent to $u_7$ via its individual key $k_7$.
This means that the adversary can obtain $k_{1-8}$,
and thus $k_{1-9}$ through the recorded traffic
corresponding to the event that $u_9$ joins the group:
\begin{eqnarray*}
\begin{array}{lll}
s \to \{u_1,\ldots,u_8\} & : & \{k_{1-9}\}_{k_{1-8}}, \{k_{789}\}_{k_{78}} \\
\end{array}
\end{eqnarray*}
As a consequence, {\em all} past and present group keys, namely $k_{1-8}$,
$k_{1-9}$, $k_{1-7,9}$ and $k_{1-5,7,9}$, are compromised even if the first three
were securely erased by $u_7$.

\item Suppose $u_5$ is compromised. Then, the adversary knows $k_{1-5,7,9}$,
$k_{45}$, and $k_5$. Further, if $k_{1-8}$ was sent to $u_5$ through
an encryption using its individual key $k_5$, then $k_{1-8}$ is exposed.
Moreover, $k_{1-9}$ can be obtained by the adversary from the
recorded traffic corresponding to the event that
$u_9$ joins the group:
\begin{eqnarray*}
\begin{array}{lll}
s \to \{u_1,\ldots,u_8\} & : & \{k_{1-9}\}_{k_{1-8}}, \{k_{789}\}_{k_{78}} \\
\end{array}
\end{eqnarray*}
As a consequence, the past and present group keys, namely $k_{1-8}$,
$k_{1-9}$ and $k_{1-5,7,9}$ are compromised, even if they
were securely erased by $u_5$.

A similar analysis applies to the case that $u_4$ is compromised.

\item Suppose $u_1$ is compromised. Then $k_{1-5,7,9}$, $k_{123}$, and $k_1$
are obtained by the adversary. This means that the adversary can further obtain
$k_{1-7,9}$ from the recorded traffic corresponding to the event that
$u_8$ leaves the group:
\begin{eqnarray*}
\begin{array}{lll}
s \to \{u_1,\ldots,u_7,u_9\} & : & \{k_{1-7,9}\}_{k_{123}},
\{k_{1-7,9}\}_{k_{456}}, \{k_{1-7,9}\}_{k_{7,9}},
\{k_{7,9}\}_{k_7}, \{k_{7,9}\}_{k_9}.
\end{array}
\end{eqnarray*}
Further, the above analysis is based on the implicit assumption that the initial
group key $k_{1-8}$ was ``magically" sent to $u_1$. In practice, $k_{1-8}$
might have been sent to $u_1$ via its individual key $k_1$.
This means that the adversary can obtain $k_{1-8}$,
and thus $k_{1-9}$ through the recorded traffic
corresponding to the event that $u_9$ joins the group:
\begin{eqnarray*}
\begin{array}{lll}
s \to \{u_1,\ldots,u_8\} & : & \{k_{1-9}\}_{k_{1-8}}, \{k_{789}\}_{k_{78}} \\
\end{array}
\end{eqnarray*}
As a consequence, {\em all} past and present group keys, namely $k_{1-8}$,
$k_{1-9}$, $k_{1-7,9}$ and $k_{1-5,7,9}$, are compromised even if the first three
were securely erased by $u_1$.

A similar analysis applies to the case
that $u_2$ or $u_3$ is compromised.
\end{itemize}

In summary, the above discussion shows, in sharp contrast to the
desired property that the adversary can only obtain the present
group key $k_{1-5,7,9}$, that compromising any of $u_1,u_2,u_3,u_7$
could lead to the exposure of all past and present group keys, and
compromising any of $u_4,u_5,u_9$ leads to the exposure of most past
and present group keys. This means that the adversary has considerable
flexibility in selecting the {\em weakest} node(s) to compromise.
Finally, we remark that the attack is not fundamentally related to
the group-oriented rekeying strategy, and that LKH+
\cite{VersaKeyJSAC99}, which was seemingly motivated from an
efficiency perspective, resolves only a piece of the problem because
the above attack remains effective when the group dynamics are
incurred by leaving events.

\begin{remark}
While there could be other methods to bootstrap the initial keys
(e.g., $k_{1-8}$ is not protected by $k_7$), the following scenario
would still support the above conclusion. Suppose at system
initialization there is no user but the server, then users join the
system one by one via LKH's join protocol (cf. Appendix
\ref{appendix:wgl}). In this case, transmission of group keys is
always protected by individual keys, meaning that compromise of some
user (or users) could lead to the exposure of all past and present
group keys.
\end{remark}

\begin{remark}
One may observe that the compromise of past group keys may not be a
serious problem. This is so because if a node stored all the
past communication content, it will be leaked to the adversary
when the node is compromised. However, there are situations, such as
sensitive applications, where the nodes do not, or even are not allowed to,
store past communication content. We
notice that this issue is also relevant to \cite{BellareCTRSA03,CanettiHKEurocrypt03}.
\end{remark}

\smallskip

\noindent{\bf Vulnerability of the One-way Function Tree (OFT) scheme:}
OFT \cite{ShermanOFT03,BalensonOFTDraft} is a stateful group communication scheme.
The basic idea underlying the OFT scheme is the following (we refer the reader to
\cite{ShermanOFT03,BalensonOFTDraft} for details).
The center maintains a binary tree, each node $x$ of
which is associated with two cryptographic keys: a node
key $k_x$ and a blinded node key $k'_x =g(k_x)$, where $g$ is an appropriate one-way function.
Every leaf of the tree is associated with a group member, and the center assigns a randomly
chosen key $k_x$ to each member $x$. Let $f$ be a ``mixing" function (e.g., $\oplus$).
The interior node keys are defined by the rule
$$
k_u = f(g(k_{left(u)}), g(k_{right(u)}))
$$
where $left(u)$ and $right(u)$ are the left and right children of the node $u$, respectively.
This way, the node key associated with the root of the tree is the group key.
In order for a member $u$ to derive the group key, the center (or server, sender)
sends the blinded node keys of
nodes adjacent to the nodes on (i.e., of the nodes ``hanging" off) the path from $u$ to the root.

When a new member joins the group, an existing leaf node $u$ is split, the member associated with
$u$ is now associated with $left(u)$, and the new member is associated with $right(u)$.
Both members are given new keys. The new blinded node keys that have been changed are
{\em securely sent} to the appropriate subgroups of members.

When the member associated with a node $u$ is evicted from the group, the member assigned to
the sibling of $u$ is reassigned to the parent $p$ of $u$ and given a new leaf key.
If the sibling $s$ of $u$ is the root of a subtree,
then $p$ becomes $s$, moving the subtree closer to the
root, and one of the leaves of this subtree is given a new key. The new blinded node keys are
{\em securely sent} to the appropriate subgroups of members.

Now we show why the OFT scheme is also vulnerable to a similar attack. The key
observation is that whenever there is a change to any blinded node
key, the center needs to {\em securely send} the new blinded node
keys to certain other legitimate nodes. It seems that any reasonable
method would be based on the keys possessed by the respective
nodes (e.g., $u$). Since $u$ can derive the new group key after
receiving the rekeying message, an {\em outsider} adversary could
use the following strategy to recover the group key: it first
records the rekeying message, and then breaks into $u$'s computer
{\em after} the next rekeying event (assuming that $u$ is still
legitimate). Moreover, compromising any of the nodes that
have not been evicted enables the adversary to recover past and
present group keys.

\smallskip

\noindent{\bf Vulnerability of the stateless subset-cover framework:}
Naor et al. \cite{NaorCrypto01} presented the first practical
stateless group communication scheme,
which has its roots in broadcast encryption \cite{FN93}.
Compared with the stateful group communication schemes discussed above, stateless schemes
have the nice feature that they do not assume the receivers
(or users, members) being always on-line.
Since the receivers do not necessarily update their state from session to session,
stateless schemes are especially good for applications over lossy channels
(e.g., MANETs and sensor networks). We stress that the security analysis
presented in \cite{NaorCrypto01} remains sound in its adversarial model;
whereas the present paper considers a strictly stronger adversarial model.

\begin{figure*}[h]
\begin{center}
\fbox{
\begin{minipage}{38pc}
\begin{description}
\item[Initialization:] Every receiver $u$ is assigned private
information $I_u$. For all $1 \leq i \leq w$
such that $u \in S_i$, $I_u$ allows $u$ to deduce the key $L_i$
corresponding to the set $S_i$.

\item[Broadcasting:] Given a set ${\cal R}$ of revoked receivers,
the center (or server, group controller,
sender) executes the following:
\begin{enumerate}
\item Choose a session encryption key $K$.
\item Find a partition of the users
in ${\cal N} \setminus {\cal R}$ into disjoint subsets $S_{i_1}, \ldots, S_{i_m}$.
Let $L_{i_1},\ldots, L_{i_m}$ be the keys associated with the above subsets.
\item Encrypt $K$ with keys $L_{i_1},\ldots, L_{i_m}$ and send the ciphertext
$$
\langle [i_1, \ldots, i_m, E_{L_{i_1}}(K), \ldots, E_{L_{i_m}}(K)], F_K(M)\rangle.
$$
\end{enumerate}
\item[Decryption:] A receiver $u$, upon receiving a broadcast message
$\langle [i_1, \ldots, i_m, C_1, \ldots, C_m], C \rangle$, executes as follows.
\begin{enumerate}
\item Find $i_j$ such that $u \in S_{i_j}$ (in the case $u \in
{\cal R}$ the result is {\sc null}). \item Extract the
corresponding key $L_{i_j}$ from $I_u$. \item Decrypt $C_j$ using
key $L_{i_j}$ to obtain $K$. \item Decrypt $C$ using key $K$ to
obtain the message $M$.
\end{enumerate}
\end{description}
\end{minipage}
}
\end{center}
\caption{The subset-cover revocation framework}
\label{fig:stateless:framework}
\end{figure*}

The subset-cover framework of \cite{NaorCrypto01} is reviewed in
Fig. \ref{fig:stateless:framework}, where
${\cal N}$ is the set of all users,
${\cal R} \subset {\cal N}$ is a group of $|{\cal R}|=r$ users whose decryption privileges
should be revoked, and $E_L$ and $F_K$ are two appropriate symmetric key cryptosystems
(whose properties will be specified later).
The goal of a stateless group communication scheme is to allow a center to transmit a message $M$
to all users such that any user $u \in {\cal N} \setminus {\cal R}$ can decrypt
the message correctly, while even a coalition consisting of all members of ${\cal R}$
cannot decrypt it.
Suppose $S_1,\ldots,S_w$ are a collection of subsets of users, where
$S_j \subseteq {\cal N}$ for $1 \leq j \leq w$, and each $S_j$ is assigned a long-lived
key $L_j$ such that each $u \in S_j$ should be able to
deduce $L_j$ from its secret information $I_u$.
Given a revoked set ${\cal R}$, if one can partition ${\cal N} \setminus {\cal R}$
into (ideally disjoint) sets $S_{i_1},\ldots,S_{i_m}$ such that
${\cal N} \setminus {\cal R} \subseteq \cup_{\ell=1}^m S_{i_\ell}$, then
a session key $K$ can be encrypted $m$ times with $L_{i_1},\ldots,L_{i_m}$, and
each user $u \in {\cal N} \setminus {\cal R}$ can obtain $K$ and thus $M$.

The subset-cover framework has a vulnerability similar to the one against
the stateful group schemes. Specifically, suppose an adversary
${\cal A} \notin {\cal N}$ records all the encrypted communications over the channels.
If ${\cal A}$ breaks into a legitimate user $u \in {\cal N}$ at a later point in time,
then $\A$ obtains $I_u$, which allows it to recover the $L_{i_j}$ (and thus the encrypted $M$)
that $u$ was entitled to obtain.
In the extreme case that $u$ has never been revoked,
${\cal A}$ can recover all past and present keys.

\ignore{
Regarding this, there are two important issues that
deserve mentioning. First, in \cite{NaorCrypto01} a weaker
security property of $F_K$ is presented. This is due to the fact
that $F_K$ uses short-lived keys (i.e., session keys). We notice
that this is a matter orthogonal to the focus of the present
paper. Second, in \cite{NaorCrypto01} a stronger security property
of $E_L$, namely security against an adaptive {\em
chosen-ciphertext} attack, is required. We observe that a secure
group communication should always be integrated with an
authentication mechanism so that the receivers are ensured that
the messages are indeed from the center. No matter how such an
authentication mechanism is implemented (see Section
\ref{sec:stateless:discussion} for more discussions),
chosen-plaintext security suffices as long as each honest receiver
only decrypts the authenticated messages. In other words,
chosen-ciphertext security is achieved by a mechanism at a
higher-layer (we stress that this does not conflict with the
rule-of-thumb cryptographic design principle that cross-layer
assurance should be avoided in general). We remark that the above
two issues are more relevant in a theoretical sense than in a
practical one.
}

\subsection{Our Contributions}

We trace the above vulnerability of group communication schemes
back to that their security models (if any) are not sufficient. We
formalize two adversarial models.
One is called the {\em passive attack model}, in which the
adversary is passive in the sense that it is only allowed to join
and leave the group in an arbitrary fashion, but not allowed to
corrupt any legitimate members. This model has seemingly been implicitly
adopted in the existing group communication literature.
The other more realistic one is called
the {\em active outsider attack model}, in which the adversary is
further allowed to corrupt legitimate members.
This model aids understanding and dealing with the aforementioned attack.
In each of the two models, we define two security notions, namely {\em
forward-security} meaning that the revoked or evicted members,
even if they collude, cannot obtain the future group keys, and
{\em backward-security} meaning that a newly admitted member
cannot obtain the past group keys.\footnote{The terms follow the
group communication literature
(see, e.g., \cite{WGL00,WallnerIETF98}).
Their meanings are indeed different from the ones adopted in
the cryptographic literature \cite{AndersonCCS97,BellareCrypto99,BellareCTRSA03}.} This allows
us to obtain the following interesting results about the
relationships between these security notions, which are equally
applicable to both stateful and stateless group communication
schemes (see Sections \ref{sec:stateful-relationships} and
\ref{sec:stateless:relationships}, respectively).
\begin{enumerate}
\item In the active outsider attack model, backward-security (also
called {\tt strong-security}) is {\em strictly} stronger
than forward-security (also called {\tt security}).
This means that in the active
outsider attack model one only needs to prove the
backward-security property.

\item In the passive attack model, backward-security is equivalent to
forward-security.

\item Backward-security in the active outsider attack model (i.e., {\tt strong-security})
is {\em strictly} stronger than backward-security in the passive attack model.
However, we do {\em not} know whether forward-security in the active outsider attack model is
also {\em strictly} stronger than its counterpart in the passive attack model
(we only know that when the adversary is {\em static} they are equivalent).

\item The security achieved in existing group communication schemes
(e.g., \cite{WGL00,ShermanOFT03,NaorCrypto01,JhoEurocrypt05})
is indeed, as we will show, forward-security in the
active outsider attack model (i.e., {\tt security}).
This has not become clear until now because there were no formal models specified before
(in spite of the fact that the passive attack model
has seemingly been implicitly adopted in the literature).
The achieved {\tt security} property is at least as
strong as what we call backward-security in the passive attack model,
but {\em strictly} weaker than what we call backward-security
in the active outsider attack model
(i.e., {\tt strong-security}) --- a property that blocks the attack discussed above.
\end{enumerate}

Besides the above general results, we show that some practical
methods can transform a {\em subclass} of the group communication
schemes (including LKH \cite{WGL00,WallnerIETF98}, LKH+
\cite{VersaKeyJSAC99}, and the complete subtree method
\cite{NaorCrypto01}) into ones that achieve the desired {\tt
strong-security}.
The methods are based on two general compilers. The
extra complexity imposed by the compilers is typically that at each rekeying
event a group member conducts logarithmically-many pseudorandom
function evaluations. This should not jeopardize their utility
even in the setting of MANETs and sensor networks, as pseudorandom
functions may be implemented with block ciphers in practice. We
also present instantiations of the compilers, which lead to
concrete schemes that achieve the desired {\tt strong-security}.

Although the technical means underlying the transformation
is to evolve the keys based on a secure pseudorandom function family --- an idea inspired by
\cite{BellareCTRSA03}, there are some subtle issues in our settings.
First, we must allow the adversary to corrupt
some group keys that are used to encrypt the communications {\em before}
the rekeying message of interest.
Of course, the corrupt members must have been revoked before that rekeying message.
On the other hand, in \cite{BellareCTRSA03} no such corruption is allowed before
the event of interest.
Second, from an adversary's perspective, there could be many ``valuable"
users in our settings, meaning that an adversary only needs to compromise the
{\em weakest} one(s) of them.
Whereas, no such flexibility is given to
the adversary in the setting of \cite{BellareCTRSA03}.

\subsection{Related Work}

LKH was independently invented in
\cite{WGL00,WallnerIETF98}.
Although these schemes are mainly invented
for secure multicast applications, we believe that many other
applications can utilize such a scheme;
we refer the reader to
\cite{RafaeliACMSurvey03,CanettiINFOCOM99} for a survey, including
the relationship between the schemes of
\cite{WGL00,WallnerIETF98} and the schemes of
\cite{FN93,NaorCrypto01}. We notice that the stateless schemes
(e.g., \cite{NaorCrypto01,JhoEurocrypt05}) are perhaps more useful in an environment
of lossy channels.
Although the LKH scheme has been extended in several directions, these extensions are
motivated to improve performance rather than to achieve strictly stronger security.
For example, performance can be improved by periodic group rekeying \cite{SetiaOkland00}
or batch rekeying \cite{LamWWW01}, and improved trade-offs between storage and
communication are available in \cite{CanettiEurocrypt99,CanettiINFOCOM99,MicciancioEurocrypt04}.
Nevertheless, these techniques may also be utilized in
our {\tt strongly-secure} group communication schemes.
To the best of our knowledge, our work is the first one
that identifies a new and realistic attack,
and presents solutions for (a subclass of) the popular group communication schemes.
The variant presented in \cite{VersaKeyJSAC99} (which is also known as LKH+)
is similar to our performance optimization in that the communication complexity
incurred by joining events can be substantially reduced.
However, there was no formal treatment of
the utilized key evolvement, nor was their scheme resistant against the attack
introduced in Section \ref{sec:motivation}.

While secure group-oriented communications have been extensively
investigated in the setting of wired networks, their counterparts
in the setting of wireless networks have yet to be thoroughly explored.
Although the aforementioned schemes can be directly deployed in wireless settings,
a simple-minded adoption may not lead to the desired performance
(see, e.g., \cite{ZhuMobiQuitous04,RadhaWINET05}).
Fortunately, there have been some interesting investigations that show that
these schemes can be adapted (e.g., by taking into account some physical characteristics
of ad hoc networks \cite{RadhaVTC04,RadhaICC04,RadhaICASSP03,RadhaWINET05})
so that better performance can be achieved.
One of the practical values of the present paper is that the significantly improved security
guarantee in the popular group communication schemes can be seamlessly integrated into
the methods for improving performance \cite{RadhaVTC04,RadhaICC04,RadhaICASSP03,RadhaWINET05}.
Indeed, our schemes can be easily integrated into any other methods for improving performance
to achieve better security,
as long as the methods assume ``black-box" access to an
underlying security group communication scheme.
There have been a few other attempts at
securing group communications in such settings:
\cite{KayaSASN03} presented a scheme for secure multicast communications in MANETs
based on public key cryptosystems; \cite{ZhuMobiQuitous04} investigated
a different approach to secure group communications.

A similar study on enhancing security of public key cryptography based
broadcast encryption was investigated in \cite{YaoCCS04}.

\subsection{Outline}

The rest of the paper is organized as follows.
In Section \ref{sec:tools} we briefly review some cryptographic preliminaries.
In Section \ref{sec:model} we present formal
models and security definitions of stateful
group communication schemes, as well as the relationships between the security notions.
In Section \ref{sec:compiler} we present a compiler for stateful group communication schemes
and investigate its properties.
The compiler is utilized in Section \ref{sec:concrete}
to derive a concrete {\tt strongly-secure} stateful group communication scheme
from the merely {\tt secure} LKH, which is reviewed in
Appendix \ref{appendix:wgl} for completeness.
In Section \ref{sec:stateless-case} we explore stateless group communication
schemes in parallel to their stateful counterparts.
We conclude the paper in Section \ref{sec:conclusion}.

\section{Cryptographic Preliminaries}
\label{sec:tools}

A function $\epsilon: {\mathbb N} \to {\mathbb R}^+$ is {\em negligible} if
$\forall c \ \exists \kappa_c \text{ s.t. } \forall \kappa>\kappa_c, \mbox{ we have }
\epsilon(\kappa)< 1/\kappa^c$.

We will base security of group
communication schemes on the security of pseudorandom function families.
For a security parameter $\kappa$, a pseudorandom function (PRF)
family $\{f_{k}\}$ parameterized by a secret value
$k \in _R \{0,1\}^{\kappa}$ has the following
property \cite{GGM86}: A probabilistic polynomial-time adversary $\A$ has
only a negligible (in $\kappa$) advantage in distinguishing $f_{k}$
from a perfect random function (with the same domain and range).
It is well-known that pseudorandom functions can be naturally used to
construct symmetric key encryption schemes that are secure against chosen-plaintext attacks
(which suffices for our treatment of LKH).
Informally, this means that no adversary is able to learn any significant information
about the encrypted content. We refer the reader to \cite{KYSTOC00} for a thorough treatment
on this subject.

\begin{definition}\emph{(computational independence)}
Consider a set $S=\{s_1,\ldots,s_\ell\}$ of secret binary strings of length $\kappa$,
where $\ell = poly(\kappa)$ for some polynomial $poly$.
We say $s_1,\ldots,s_\ell$ are {\em computationally independent} of each other if
for any probabilistic polynomial-time algorithm $\A$,
$$
\left|\Pr[ \A(s_1, \ldots, s_\ell) \text{ returns ``real" }]-
\Pr[\A(r_1, \ldots, r_\ell) \text{ returns ``real" }] \right|= \epsilon(\kappa)
$$
where $r_1,\ldots,r_\ell \in _R\{0,1\}^\kappa$ are uniformly drawn at random,
and $\epsilon$ is a negligible function.
\end{definition}

\section{Model and Definition of Stateful Secure Group Communications}
\label{sec:model}

In Section \ref{sec:s-model1} we present a formal security model
for stateful (and symmetric key cryptography based) group
communications. In Section \ref{sec:s-model2} we specify the
adversarial models and desired security properties.
In Section \ref{sec:stateful-relationships} we explore the relationships between
the security notions.

\subsection{Model}
\label{sec:s-model1}

Let $\kappa$ be a security parameter, and $\ID$ be the set of possible group members
(i.e., users, receivers, or principals)
such that $|\ID|$ is polynomially-bounded in $\kappa$.
There is a special entity called a {\em Group Controller}
(i.e., key server, center, server, or sender),
denoted by $\GC$, such that $\GC \notin \ID$.

Since a stateful group communication scheme is driven
by ``rekeying" events (because of joining or leaving
operations below), it is
convenient to treat the events as occurring at ``virtual time" $t=0,1,2,\ldots$, because
the group controller is able to maintain such an execution history. This indeed accommodates
the following important two cases:
(1) all the parties periodically update their keys,
even if there are no joining or leaving operations
--- this is relevant when a scheme achieves what
we call {\tt strong-security} specified below;
(2) the lengths of the time periods do not have to be the same ---
this is the case when the rekeying events occur in an arbitrary
fashion. At time $t$, let $\Delta^{(t)}$ denote the set of
legitimate group members,
$k^{(t)}=k_{\GC}^{(t)}=k_{U_1}^{(t)}=\ldots$ the group
key,\footnote{It is also known as a session key in the group
communication literature.} $K_{\GC}^{(t)}$ the set of keys held by
the $\GC$, $K_U^{(t)}$ the set of keys held by $U \in \Delta^{(t)}$,
and $\acc_U^{(t)}$ the state indicating whether $U \in \Delta^{(t)}$
has successfully received the rekeying message. Initially,
$\forall$ $U \in \ID, t \in \N$, set $\acc_U^{(t)} \leftarrow
\FALSE$. We assume that the $\GC$ treats joining and leaving operation
separately (e.g., first fulfilling the leaving operation and then
immediately the joining one), even if the requests are made
simultaneously. This strategy has indeed been adopted in the group
communication literature.

To simplify the presentation, we assume that during the system
initialization (i.e., $\setup$ below) or the admission of a joining
user, the $\GC$ can communicate with each legitimate member $U \in \ID$
through an {\em authenticated private} channel, and that after the
system initialization the $\GC$ can communicate with any $U$ through an
{\em authenticated} channel. We notice that authenticated channels
can by implemented by a digital signature scheme \cite{WGL00}, and
digital signatures are sometimes necessary \cite{BonehEurocrypt01}.

We will not make any {\em synchronization}
assumption about the underlying communication model, which
could therefore be asynchronous \cite{Lynch96}. However, known practical schemes
(e.g., \cite{WGL00,WallnerIETF98,CanettiEurocrypt99}) assume {\em reliable} delivery, which
would require some (loose) clock synchronization.

A group communication scheme has the following components:

\begin{description}
\item[{\sf Setup}:] The group controller $\GC$ generates a set of keys
$K_{\GC}^{(0)}$, and distributes appropriate subsets of $K_{\GC}^{(0)}$ to the present
group members (that may be determined by the adversary),
$\Delta^{(0)} \subseteq \ID$, through the authenticated private channels.
Each member $U_i \in \Delta^{(0)}$ holds a set of keys denoted by
$K_{U_i}^{(0)} \subset K_{\GC}^{(0)}$, and there is a key, $k^{(0)}$
that is common to all the present members, namely
$k^{(0)} \in K_{\GC}^{(0)} \cap K_{U_1}^{(0)} \cap \ldots \cap K_{U_{|\Delta^{(0)}|}}^{(0)}$.

\item[{\sf Join}:] This algorithm is executed by group controller $\GC$ at
time, say, $t$ due to some join request(s) (we abstract away
the out-of-band authentication and establishment of an individual
key for each of the new members). It takes as input: (1)
identities of previous group members, $\Delta^{(t-1)}$, (2) identities of
newly admitted group members, $\Delta' \subseteq \ID \setminus
\Delta^{(t-1)}$, (3) keys held by the group controller, $K_{\GC}^{(t-1)}$, and (4)
keys held by group members,
$\{K_{U_i}^{(t-1)}\}_{U_i \in \Delta^{(t-1)}} = \{K_{U_i}^{(t-1)} : U_i \in \Delta^{(t-1)} \}$.

It outputs updated system state information, including:
(1) identities of new group members,
$\Delta^{(t)} \leftarrow \Delta^{(t-1)} \cup \Delta'$,
(2) new keys for the $\GC$ itself, $K_{\GC}^{(t)}$,
(3) new keys for new group members,
$\{K_{U_i}^{(t)}\}_{U_i \in \Delta^{(t)}}$,
which are sent to the legitimate users
through the authenticated channels,
(4) new group key
$k^{(t)} \in K_{\GC}^{(t)} \cap K_{U_1}^{(t)} \cap \ldots \cap K_{U_{|\Delta^{(t)}|}}^{(t)}$.

Formally, denote it by
$(\Delta^{(t)}, K_{\GC}^{(t)}, \{K_{U_i}^{(t)}\}_{U_i \in \Delta^{(t)}})
\leftarrow {\sf Join}( \Delta^{(t-1)}, \Delta',
K_{\GC}^{(t-1)}, \{K_{U_i}^{(t-1)}\}_{U_i \in \Delta^{(t-1)}})$.

\item[{\sf Leave}:]  This algorithm is executed by the group controller $\GC$ at
time, say, $t$ due to leave or revocation operation(s).
It takes as input:
(1) identities of previous group members, $\Delta^{(t-1)}$,
(2) identities of leaving group members, $\Delta' \subseteq \Delta^{(t-1)}$,
(3) keys held by the controller, $K_{\GC}^{(t-1)}$, and
(4) keys held by group members, $\{K_{U_i}\}_{U_i \in \Delta^{(t-1)}}^{(t-1)}$.

It outputs updated system state information, including:
(1) identities of new group members,
$\Delta^{(t)} \leftarrow \Delta^{(t-1)} \setminus \Delta'$,
(2) new keys for $\GC$, $K_{\GC}^{(t)}$,
(3) new keys for new group members,
$\{K_{U_i}^{(t)}\}_{U_i \in \Delta^{(t)}}$, which
are sent to the legitimate users
through the authenticated channels,
(4) a new group key
$k^{(t)} \in K_{\GC}^{(t)} \cap K_{U_1}^{(t)} \cap \ldots \cap K_{U_{|\Delta^{(t)}|}}^{(t)}$.

Formally, denote it by
$(\Delta^{(t)}, K_{\GC}^{(t)}, \{K_{U_i}^{(t)}\}_{U_i \in \Delta^{(t)}})
\leftarrow {\sf Leave}( \Delta^{(t-1)}, \Delta',
K_{\GC}^{(t-1)}, \{K_{U_i}^{(t-1)}\}_{U_i \in \Delta^{(t-1)}})$.

\item[{\sf Rekey}:] This algorithm is executed by the legitimate
group members belonging to $\Delta^{(t)}$ at time $t$,
where $\Delta^{(t)}$ is derived from a {\sf Join} or {\sf Leave}
event. Specifically, $U_i \in \Delta^{(t)}$ runs this algorithm
upon receiving the message from the $\GC$ over the authenticated
channel. The algorithm takes as input the received message and
$U_i$'s secrets, and is supposed to output the updated keys for
the group member. If the execution of the algorithm is successful,
$U_i$ sets: (1) ${\sf acc}_{U_i}^{(t)} \leftarrow \TRUE$, (2)
$K_{U_i}^{(t)}$, where $k_{U_i}^{(t)} \in K_{U_i}^{(t)}$ is
supposed to be the new group key.

If the rekeying event is incurred by a {\sf Join} event,
every $U_i \in \Delta^{(t)}$ erases $K_{U_i}^{(t-1)}$ and any
temporary storage after obtaining $K_{U_i}^{(t)}$. If the rekeying event is incurred by
a {\sf Leave} event, every $U_i \in \Delta^{(t)}$
erases $K_{U_i}^{(t-1)}$ and any temporary storage after
obtaining $K_{U_i}^{(t)}$, and every {\em honest}
leaving group member $U_j \in \Delta'$ erases $K_{U_j}^{(t-1)}$
(although a {\em corrupt} one does not have to follow this protocol).
\end{description}

We require that any group communication scheme satisfy
the following {\tt correctness} requirement: for any $t=1,2,\ldots$ and
$\forall$ $U \in \Delta^{(t)}$, if
$\acc_U^{(t)} = \TRUE$, then $k_{U}^{(t)} = k^{(t)}$ and $K_{U}^{(t)} \subset K_{\GC}^{(t)}$.

\subsection{Security Definitions}
\label{sec:s-model2}

We consider an adversary that has complete control over all the
communications in the network. To simplify the definition, we
assume that the group controller is never compromised; this is not
necessarily a restriction because the adversary could have
compromised all the group members (and thus have obtained the
secrets the group controller holds).

An adversary's interaction with principals in the network is modeled
by allowing it to have access to (some of) the following oracles:

\begin{itemize}
\item $\Ora_{Send}(U,t,M,action)$: Send a message $M$ to $U \in \{\GC\} \cup \ID$
at time $t\geq 0$, and output its reply, where
$action \in \{{\sf Setup}, {\sf Join}, {\sf Leave}, {\sf Rekey}\}$ meaning that
$U$ will execute according to the corresponding protocol, and
$M$ specifies the needed information for executing the protocol.
Of course, the query of type {\sf Setup} is only made at time $t=0$.

These oracle accesses are meant to capture that the adversary can observe the reactions
of the non-corrupt participants (e.g., the incurred message exchanges).
For example, the adversary can let some
honest (i.e., non-corrupt) users join or leave the group in question.

\item $\Ora_{Reveal}(U,t)$: Output the group key held by $U \in  \Delta^{(t)}$
at time $t$, namely $k_U^{(t)}$.

\item $\Ora_{Corrupt}(U,t)$: Output the keys held by $U \in \Delta^{(t)}$ at time $t$, namely
$K_U^{(t)}$.

\item $\Ora_{Test}(U,t)$: This oracle may be queried only once, at any time
during the adversary's execution. A random bit $b$ is generated:
if $b=1$ the adversary is given $k_U^{(t)}$ where $U \in \Delta^{(t)}$,
and if $b=0$ the adversary is given a random key of length $|k_U^{(t)}|$.
\end{itemize}

Now we define the {\em active outsider attack model} that is strictly more powerful than
the {\em passive outsider attack model} that has been implicitly utilized in the literature.

\begin{definition}\emph{(active outsider attack model)}
\label{def:adversarial-model} In this model, the adversary $\A$
may have access to all the oracles specified above. In particular,
an ``outsider" $\A \notin \Delta^{(t)}$ is allowed to issue an
$\OO_{Reveal}(U,t)$ or $\OO_{Corrupt}(U,t)$ query for some $U \in
\Delta^{(t)}$.
\end{definition}

\begin{definition}\emph{(passive attack model)}
\label{def:adversarial-model}
In this model, the adversary is only allowed to make $\Ora_{Send}(\cdot,\cdot,\cdot,\cdot)$
and $\Ora_{Test}(\cdot,\cdot)$ queries.
In other words, the adversary is only allowed to join and leave the group (in an
arbitrary fashion though).
\end{definition}

In each of the two models, we define two security notions:
backward-security and forward-security.
This leads to four security notions:
(1) forward-security in the active outsider attack model or simply {\tt security},
(2) backward-security in the active outsider attack model or simply {\tt strong-security},
(3) forward-security in the passive attack model, and
(4) backward-security in the passive attack model.

\begin{definition}\emph{({\tt security})}
\label{def:stateful-forward-secrecy}
Intuitively, it means that $\A$ learns no
information about a group key if (1) with respect to the corresponding rekeying event
there is no corrupt legitimate member (this implicitly implies that
all the members that were corrupted by $\A$ must have been revoked), and
(2) no member is corrupted by $\A$ after the rekeying event.
Formally, consider the following event {\sf Succ}:
\begin{description}
\item[(1)] The adversary can make arbitrary oracle queries at any time
$t_1<t$, except the following restrictions hold.

\item[(2)] The adversary queries the $\Ora_{Test}(U,t)$ oracle
with ${\sf acc}_U^{(t)} =\TRUE$,
and correctly guesses the bit $b$ used by the $\Ora_{Test}(U,t)$
oracle in answering this query.

\item[(3)] There is no $\Ora_{Reveal}(V,t)$ query for any
$V \in \Delta^{(t)}$.
(Otherwise, the group key is trivially compromised.)

\item[(4)] For every $\Ora_{Corrupt}(V,t_1)$ query where $t_1 < t$,
there must have been an $\Ora_{Send}(\GC,t_2,V,{\sf Leave})$ query where
$t_1 < t_2 \leq t$.
This captures that the corrupt members must have been revoked
before the rekeying message at time $t$.

\item[(5)] There is no $\Ora_{Corrupt}(V,t_3)$ query for
any $t_3 \geq t$ and $V \in \Delta^{(t_3)}$.
\end{description}
The advantage of the adversary $\A$ in attacking the group communication scheme is defined as
${\sf Adv}_{\A}(\kappa) = | 2 \cdot \Pr[{\sf Succ}] -1|$,
where $\Pr[{\sf Succ}]$ is the probability that the event {\sf Succ} occurs, and the probability
is taken over the coins used by $\GC$ and by $\A$.
We say a scheme is {\tt secure} if for all
probabilistic polynomial-time adversary $\A$ it holds that
${\sf Adv}_{\A}(\kappa)$ is negligible in $\kappa$.
\end{definition}

\begin{definition}\emph{({\tt strong-security})}
\label{def:stateful-backward-secrecy}
Intuitively, it means that an adversary learns no
information about a group key if, with respect to the rekeying event of interest
there is no corrupt legitimate member (this implicitly implies that
all the previously corrupt members have been revoked).
Formally, consider the following event {\sf Succ}:
\begin{description}
\item[(1)-(4)] The same as in the definition of {\tt security}.

\item[(5)]
There is no $\Ora_{Corrupt}(V,t_3)$ query for $t_3=t$ and $V \in \Delta^{(t_3)}$.
(This does not rule out that there could be some $\Ora_{Corrupt}(V,t_3)$ query for $t_3 > t$.)
\end{description}
The advantage of the adversary $\A$ in attacking the group communication scheme is defined as
${\sf Adv}_{\A}(\kappa) = | 2 \cdot \Pr[{\sf Succ}] -1|$,
where $\Pr[{\sf Succ}]$ is the probability that the event {\sf Succ} occurs, and the probability
is taken over the coins used by $\GC$ and by $\A$.
We say a scheme is {\tt strongly-secure} if for all
probabilistic polynomial-time adversary $\A$ it holds that
${\sf Adv}_{\A}(\kappa)$ is negligible in $\kappa$.
\end{definition}

\begin{definition}\emph{(forward-security in the passive attack model)}
\label{def:stateful-forward-secrecy-passive}
Intuitively, it means that $\A$, which is not allowed to make any $\OO_{Reveal}$ or
$\OO_{corrupt}$ query, learns no
information about any group key after leaving the group.
Formally, consider the following event {\sf Succ}:
\begin{description}
\item[(1)] The adversary arbitrarily queries the
$\Ora_{Send}(\cdot,t_1,\cdot,\cdot)$ oracle for any $t_1 < t$.
Moreover, the adversary itself can arbitrarily join or leave the group at time $t_1$, provided
that the following restriction holds.

\item[(2)] The adversary queries the $\Ora_{Test}(U,t)$ oracle, where
(1) ${\sf acc}_U^{(t)} =\TRUE$ for an honest user $U$, and (2)
$\A \notin \Delta^{(t)}$.
Then, the adversary correctly guesses the bit $b$ used by the $\Ora_{Test}(U,t)$
oracle in answering this query.
\end{description}
The advantage of the adversary $\A$ in attacking the group communication scheme is defined as
${\sf Adv}_{\A}(\kappa) = | 2 \cdot \Pr[{\sf Succ}] -1|$,
where $\Pr[{\sf Succ}]$ is the probability that the event {\sf Succ} occurs, and the probability
is taken over the coins used by $\GC$ and by $\A$.
We say a scheme is {\tt secure} if for all
probabilistic polynomial-time adversary $\A$ it holds that
${\sf Adv}_{\A}(\kappa)$ is negligible in $\kappa$.
\end{definition}

\begin{definition}\emph{(backward-security in the passive attack model)}
\label{def:stateful-backward-secrecy-passive}
Intuitively, it means that $\A$, which is not allowed to make any $\OO_{Reveal}$ or
$\OO_{corrupt}$ query, learns no
information about any group key before joining the group (again).
Formally, consider the following event {\sf Succ}:
\begin{description}
\item[(1)] The adversary may arbitrarily query the
$\Ora_{Send}(\cdot,t_1,\cdot,\cdot)$ oracle for any $t_1 < t$.
Moreover, $\A$ can arbitrarily join or leave the group at time $t_1$, provided
that the following restriction holds.

\item[(2)] The adversary queries the $\Ora_{Test}(U,t)$ oracle, where
(1) ${\sf acc}_U^{(t)} =\TRUE$ for an honest user $U$, and (2)
$\A \notin \Delta^{(t)}$.

\item[(3)] The adversary queries the $\Ora_{Send}(\cdot,t_2,\cdot,\cdot)$ oracle for any $t_2 > t$.
Moreover, $\A$ can arbitrarily join or leave the group at time $t_2$.

\item[(4)] The adversary correctly guesses the bit $b$ used by the $\Ora_{Test}(U,t)$
oracle in answering this query.
\end{description}
The advantage of the adversary $\A$ in attacking the group communication scheme is defined as
${\sf Adv}_{\A}(\kappa) = | 2 \cdot \Pr[{\sf Succ}] -1|$,
where $\Pr[{\sf Succ}]$ is the probability that the event {\sf Succ} occurs, and the probability
is taken over the coins used by $\GC$ and by $\A$.
We say a scheme is {\tt secure} if for all
probabilistic polynomial-time adversary $\A$ it holds that
${\sf Adv}_{\A}(\kappa)$ is negligible in $\kappa$.
\end{definition}

It is trivial to see that {\tt strong-security} implies
{\em backward-security in the passive model}, and that
{\tt security} implies {\em forward-security in the passive attack model}.

\subsection{Relationships between the Security Notions}
\label{sec:stateful-relationships}

We summarize the relationships between the security notions of
stateful group communication schemes in Fig.
\ref{fig:relationships}, where $X \rightarrow Y$ means $X$ is
stronger than $Y$, $X \leftrightarrow Y$ means $X$ is equivalent
to $Y$, $X \not\rightarrow Y$ means $X$ does not imply $Y$, and $X
\stackrel{?}{\not\rightarrow} Y$ means it is unclear where $X$
does not imply $Y$.
Below we elaborate on the non-trivial relationships showed in Fig. \ref{fig:relationships}.

\begin{figure*}[h]
\unitlength=1.00mm \special{em:linewidth 0.4pt} \linethickness{0.4pt}
\begin{picture}(132.00,40.00)
\put(16.00,38.00){{backward-security in the}}
\put(16.00,34.00){{active outsider attack model}}
\put(16.00,30.00){{(i.e., {\tt strong-security})}}
\put(15.00,29.00){\framebox(46, 12)}
\put(68.00,40.00){{Proposition \ref{proposition:active-model:backward-implies-forward}}}
\put(61.00,39.00){\vector(1,0){38}}
\put(68.00,35.00){{Proposition \ref{proposition:active-model:backward-does-not-imply-forward}}}
\put(78.00,32.00){/}
\put(99.00,33.00){\vector(-1,0){38}}
\put(16.00,8.00){{backward-security in the}}
\put(16.00,2.00){{passive attack model}}
\put(15.00,0.00){\framebox(46, 12)}
\put(14.00,20.00){trivial}
\put(24.00,29.00){\vector(0,-1){17.00}}
\put(54.00,20.00){Theorem \ref{theorem:counter-example}}
\put(52.00,21.00){/}
\put(53.00,12.00){\vector(0,1){17.00}}
\put(100.00,38.00){{forward-security in the}}
\put(100.00,34.00){{active outsider attack model}}
\put(100.00,30.00){{(i.e., {\tt security})}}
\put(99.00,29.00){\framebox(46, 12)}
\put(68.00,7.00){{Proposition \ref{proposition:passive-model:backward-equals-to-forward}}}
\put(81.00,6.00){\vector(1,0){18}}
\put(81.00,6.00){\vector(-1,0){20.00}}
\put(100.00,8.00){{forward-security in the}}
\put(100.00,2.00){{passive attack model}}
\put(99.00,0.00){\framebox(46, 12)}
\put(98.00,20.00){trivial}
\put(108.00,29.00){\vector(0,-1){17.00}}
\put(136.00,21.00){{/}}
\put(139.00,18.00){{\huge ?}}
\put(137.00,12.00){\vector(0,1){17.00}}
\end{picture}
\caption{The relationships between the security notions in stateful group communication schemes}
\label{fig:relationships}
\end{figure*}
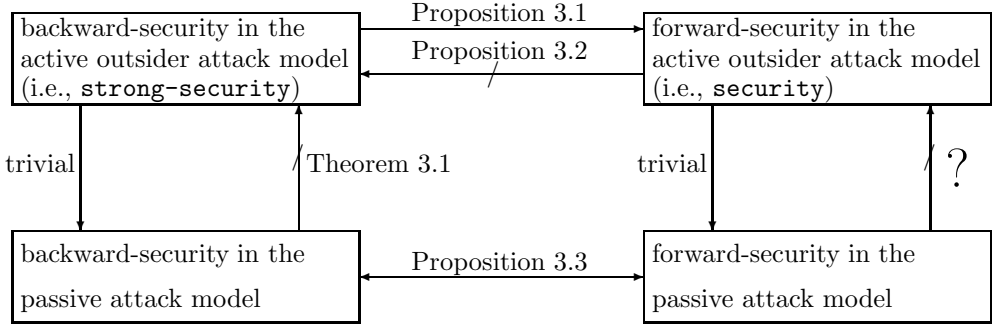

\begin{proposition}
\label{proposition:active-model:backward-implies-forward}
If a stateful group communication scheme is {\tt strongly-secure}, then it is also {\tt secure}.
\end{proposition}

\begin{proof}
This is almost immediate because, on one hand, the definition of {\tt strong-security}
ensures the secrecy of $k^{(t)}$ even if $\A$ corrupts some $U \in \Delta^{(t_3)}$
where $t_3 > t$, and on the other hand,
the definition of {\tt security}
ensures the secrecy of $k^{(t)}$ only if $\A$ does not corrupt any $U \in \Delta^{(t_3)}$
for any $t_3>t$.  \qed
\end{proof}

\begin{proposition}
\label{proposition:active-model:backward-does-not-imply-forward} A
stateful group communication scheme that is {\tt secure} is not
necessarily {\tt strong-secure}.
\end{proposition}

\begin{proof}
The fact that {\tt security} does not imply {\tt strong-security} is implied by Theorem
\ref{theorem:wgl}, which states that LKH is {\tt secure}, and that
LKH is insecure against an active outside attacker (cf. the attack scenario in
Section \ref{sec:motivation}).
\qed
\end{proof}

The above proposition implies that for a stateful group communication scheme, one only needs to
show that it is {\tt strongly-secure}.

\begin{proposition}
\label{proposition:passive-model:backward-equals-to-forward}
A stateful group communication scheme is backward-secure in the passive attack model
iff it is forward-secure in the passive attack model.
\end{proposition}

\begin{proof}
First we show that a group communication scheme that is not forward-secure in the passive attack
model is also not backward-secure in the passive attack model.
 Suppose $\A$ first joins the group at time $t_1$ and then leaves the group
at time $t_2$ where $t_1 < t_2$. Since the scheme is not forward-secure in the passive attack
model, $\A$ can distinguish $k^{(t_3)}$ from a random string for
some $t_3 > t_2$ with a non-negligible probability.
Now suppose $\A$ re-joins the group at time $t_4$ for some $t_4 >t_3$.
Then, with respect to this re-joining event, $\A$
can distinguish $k^{(t_3)}$ from a random string with a non-negligible probability.
Since $\A$ did not make any $\OO_{Reveal}$ or $\OO_{corrupt}$ query,
the scheme is not backward-secure in the passive attack model.

Now we show that a group communication scheme that is
not backward-secure in the passive attack model is also not
forward-secure in the passive attack model.
Suppose $\A$ first joins the group at time $t_1$, leaves the group
at time $t_2$, and re-joins the group at time $t_3$, where
$t_1 < t_2 < t_3 $. Since the scheme is not backward-secure in the passive attack model,
$\A$ can distinguish $k^{(t)}$ from a random string for some $t_2 \leq t < t_3$
with a non-negligible probability.
This also means that, with respect to the leaving event at time $t_2$, $\A$
can distinguish $k^{(t)}$ from a random string for some $t \geq t_2$ with a
non-negligible probability.
Since $\A$ does not make any $\OO_{Reveal}$ or $\OO_{corrupt}$ query,
the scheme is not forward-secure in the passive attack model.
\qed
\end{proof}

We do not know whether forward-security in the passive attack mode
also implies forward-security in the active outsider attack model.
The relationship may seem trivial at a first glance, since all the
corrupt members are revoked before the ``challenge" session, and
the adversary is not allowed to corrupt any member after the
``challenge" session. Although it can indeed be shown that the
implication holds, provided that the adversary is {\em static}
(meaning that the adversary decides which principals in $\ID$ it
will corrupt at system initialization), in the more interesting
case that the adversary is {\em adaptive}, we do not know how to
prove it.

\begin{theorem}
\label{theorem:counter-example}
There exists a group communication scheme that is backward-secure in the passive attack model
but not {\tt strongly-secure} (i.e., backward-secure in the active outsider attack model).
\end{theorem}

\begin{proof}
Theorem \ref{theorem:wgl} shows that LKH is {\tt secure}
(i.e., forward-secure in the active outsider attack model), which trivially means that
it is also forward-secure in the passive attack model.
Then, Proposition \ref{proposition:passive-model:backward-equals-to-forward}
shows that it is also backward-secure in the passive attack model.

On the other hand, the attack scenario shown in Section
\ref{sec:motivation} states that LKH is not
backward-secure in the active outsider attack model. \qed
\end{proof}

\section{A Compiler for Stateful Group Communication Schemes}
\label{sec:compiler}

Suppose $\{f_k \}$ is a secure pseudorandom function family.
Now we present a compiler that transforms a {\tt secure} group communication scheme,
$\SGC=(\setup,\join,\leave,\rekey)$,
into a {\tt strongly-secure} one, $\SSGC=(\setup^*,\join^*,\leave^*,\rekey^*)$.
The compiler applies to the subclass of stateful group communication schemes
where the different keys belong to $K_{\GC}^{(t)}$
are computationally independent of each other, where $t=0,1,2,\ldots$.
In what follows ``a key $k$ needs to be changed" means that it should be substituted with
a random key that is information-theoretically independent of $k$.

The key idea behind the compiler is to update the keys, which are possibly
used to encrypt the new keys that need to be securely sent to the legitimate
users, at each join, leave, or rekey event via an appropriate family of
pseudorandom functions. As a result, compromise of a current key does not
allow the adversary to recover the corresponding past keys.

\begin{description}
\item[${\sf Setup}^*$:] This is the same as $\SGC.\setup$.

\item[${\sf Join}^*$:] This algorithm is executed by $\GC$ at time, say, $t$.
Let $K$ be the set of keys that need to be
changed (including the group key $k^{(t-1)}$),
$K^*$ be the set of common key(s) shared between the $\GC$ and the joining user(s),
and $K^{**}$ be the new keys (including the new
group key $k^{(t)}$) that are used to replace the keys in $K$.

\begin{enumerate}
\item Execute $\SGC.\join$ except for the following:
(1) for every $k_i \in (K_{\GC}^{(t-1)} \setminus \{k^{(t-1)}\}) \cup K^*$,
let $f_{k_i}(0)$ play the role of $k_i$ in $\SGC.\join$;
(2) for every $k_i \in K^{**} \setminus \{k^{(t)}\}$ that is used as an encryption key in $\SGC.\join$,
let $f_{k_i}(0)$ play the role of $k_i$.

\item Every individual key
$k_i \in (K_{\GC}^{(t-1)} \setminus K) \cup K^* \cup (K^{**} \setminus \{k^{(t)}\})$
is replaced by $f_{k_i}(1)$.
\end{enumerate}

\item[${\sf Leave}^*$:] This algorithm is executed by $\GC$ at time, say, $t$.
Let $K$ be the set of keys that need to be
changed (including the group key $k^{(t-1)}$) or eliminated,
and $K^{**}$ be the new keys (including the new
group key $k^{(t)}$) that are used to replace (possibly a subset of) the keys in $K$.

\begin{enumerate}
\item Execute $\SGC.\leave$ except for the following:
(1) for every $k_i \in K_{\GC}^{(t-1)} \setminus K$,
let $f_{k_i}(0)$ play the role of $k_i$ in $\SGC.\leave$;
(2) for every $k_i \in K^{**} \setminus \{k^{(t)}\}$ that is used
as an encryption key in $\SGC.\leave$,
let $f_{k_i}(0)$ play the role of $k_i$.

\item Every individual key
$k_i \in (K_{\GC}^{(t-1)} \setminus K) \cup (K^{**} \setminus \{k^{(t)}\})$
is replaced by $f_{k_i}(1)$.
\end{enumerate}

\item[${\sf Rekey}^*$:] There are two cases.
\begin{itemize}
\item The rekeying event is incurred by a {\sf Leave} event at time $t$.
In this case, every honest leaving user should erase all the secrets as in $\SGC.\rekey$,
and every remaining user, $V \in \Delta^{(t)}$, executes the following.
Denote by $K^{'}_V \subseteq K_V^{(t-1)}$ the subset of keys
that need to be changed to a set of new keys $K^{''}_V$.
(We notice that both $K^{'}_V$ and $K^{''}_V$ can be derived by $V$
after receiving the rekeying message
and that $k^{(t)} \in K^{''}_V$.)
First, $V$ executes $\SGC.\rekey$ except for
letting $f_{k_i}(0)$ play the role of $k_i$
under the circumstance that $k_i \in K_{V}^{(t-1)} \setminus \{k^{(t-1)}\}$
or $k_i \in K^{''}_V \setminus \{k^{(t)}\}$ is used as an encryption key,
and updates every
$k_i \in (K_V^{(t-1)} \setminus K^{'}_V) \cup (K^{''}_V \setminus \{k^{(t)}\})$
as $f_{k_i}(1)$.
Second, $V$ erases the outdated keys (except $K_V^{(t)}$) as in $\SGC.\rekey$.

\item The rekeying event is incurred by a {\sf Join} event at time $t$.
We notice that user $V \in \Delta^{(t)}$ holds a set of keys $K_V^{(t-1)}$
(in the case of $V$ being a joining user,
$K_V^{(t-1)}$ consists of the only common key between $\GC$ and $V$),
of which a subset $K^{'}_V$ of keys (which may be empty)
are to be changed to a set of new keys $K^{''}_V$.
(We notice that both $K^{'}_V$ and $K^{''}_V$ can be derived by $V$
after receiving the rekeying message
and that $k^{(t)} \in K^{''}_V$.)
First, $V$ executes $\SGC.\rekey$ except for
letting $f_{k_i}(0)$ play the role of $k_i$
under the circumstance that $k_i \in K_{V}^{(t-1)} \setminus \{k^{t-1}\}$ or
$k_i \in K^{''}_V \setminus \{k^{(t)}\}$ is used as an encryption key, and
updates every $k_i \in (K_V^{(t-1)} \setminus K^{'}_V) \cup (K^{''}_V \setminus \{k^{(t)}\})$
as $f_{k_i}(1)$.
Second, $V \in \Delta^{(t)}$ erases the outdated keys
(other than $K_V^{(t)}$) as in $\SGC.\rekey$.
\end{itemize}
\end{description}

\subsection{Analysis}

First we analyze the complexity of $\SSGC$.
\begin{itemize}
\item It does not introduce any extra communication complexity over $\SGC$;
this is important in many applications such as MANETs
and sensor networks. (In Section \ref{sec:improvement} we further reduce the communication
complexity.)
\item It does not introduce any extra storage complexity over $\SGC$, provided that
the temporary storage for the keys such as $f_{k_i}(0)$ is
insignificant. This is at least true for most applications
including MANETs and sensor networks.

\item The only extra complexity of $\SSGC$ over $\SGC$ is the
evaluation of the pseudorandom functions. Specifically, the server
needs to conduct $O(\max\{|K_{\GC}^{(t-1)}|, |K_{\GC}^{(t)}|\})$
pseudorandom function evaluation operations;
a user $V$ needs to conduct $O(\max\{|K_{U}^{(t-1)}|, |K_{U}^{(t)}|\})$
pseudorandom function operations. We notice that
typically $|K_{U}^{(t)}| = O(\log(|K_{\GC}^{(t)}|))$ (e.g., \cite{WGL00}). This should be
insignificant for most applications including MANETs and sensor
networks.
\end{itemize}

Now we prove that $\SSGC$ is indeed {\tt strongly-secure}.
The intuition that $\SSGC$ is {\tt strongly-secure} (and thus defeats
the attacks presented in the Introduction) is due to the following fact:
compromise of all of the current keys held by a user does not necessarily allow
the adversary to recover any of the past keys, which may have been used
to secure the transimission of other keys.

\begin{theorem}
\label{theorem:main1}
Assume $\{f_k\}$ is a secure pseudorandom function family (as
specified in Section \ref{sec:tools}).
If $\SGC$ is {\tt secure}, then $\SSGC$ is {\tt strongly-secure}.
\end{theorem}

\begin{proof} (sketch) We show that if $\SSGC$ is not {\tt strongly-secure}, then
$\SGC$ is not {\tt secure}. Note that the key difference between the two security notions
is whether the adversary is allowed to corrupt a legitimate
user after the rekeying event of interest.
Note also that after a rekeying event in $\SSGC$ all the new
keys are either information-theoretically
or computationally independent of each other.

First, consider a mental scheme that is the same as in $\SSGC$ except that
the keys of the form $f_{k_i}(0)$
are always substituted with freshly and independently
chosen random keys, where $k_i$ is not held by any corrupt user.
We claim that this mental scheme achieves {\tt strong-security};
otherwise, there is an efficient
algorithm to break the {\tt security} of $\SGC$.
To see this, we construct a simulator that has access to a challenge $\SGC$ environment.
Since the number of rekeying events in $\SSGC$
is polynomially-bounded, the simulator has an inverse
polynomial probability in successfully guessing
the rekeying event of interest -- the event
corresponding to the $\OO_{Test}(\cdot,\cdot)$ query.
\begin{enumerate}
\item The simulator interacts with the adversary as in $\SSGC$; this can be done because
the simulator has complete control over the keys utilized in the $\SSGC$.
We notice that the simulator can answer any queries,
including $\OO_{reveal}$ and $\OO_{corrupt}$.

\item When the simulated $\SSGC$ execution reaches the point of interest,
the simulator asks the challenge $\SGC$ environment to establish an
instance of $\SGC$ with the same
set of users. The establishment of the instance is
based on the rekeying event incurred by the adversary in $\SSGC$,
so that the legitimate users hold the corresponding keys as in $\SGC.{\sf Setup}$.
This substitution can get through because the definition of {\tt strong-security}
ensures that the adversary in $\SSGC$ does not corrupt any
legitimate user during the rekeying event.

\item At the next rekeying event, the simulator can continue its execution of
the $\SSGC$ because it can utilize independent secrets that are freshly chosen by itself.
We notice that the simulator can answer any queries,
including $\OO_{reveal}$ and $\OO_{corrupt}$, as it can
simulate the $\SSGC$ environment in any future rekeying events.
\end{enumerate}

Second, it is clear that the difference between $\SSGC$ and the
aforementioned mental scheme is that the keys utilized in the rekeying events are
either series of keys in the forms of $f_{k_i}(1)$ where
the $k_i$'s are secret from the adversary,
or freshly and independently chosen at random.
We claim that the two cases are indistinguishable
as long as the pseudorandom function family is secure.
To see this, we notice that no adversary can, with a non-negligible probability,
distinguish a single key-chain of a fixed key identity,
namely $f_{k_i}(1)$, $f_{f_{k_i}(1)}(1)$, $f_{f_{f_{k_i}(1)}(1)}(1), \ldots$ where
$k_i$ is secret from the adversary, from a sequence of random secrets.
Otherwise, we can construct an algorithm to distinguish a pseudorandom function from a random one
with a non-negligible probability (because the number of rekeying events is polynomially-bounded).
Conditioned on the fact that the number of keys is polynomially-bounded,
we conclude that the keys derived from pseudorandom functions are indistinguishable from
the keys that are freshly and independently chosen;
otherwise, a standard hybrid argument shows that there
exists an algorithm that is able to distinguish a pseudorandom function from a random one
(because there are at most a polynomial number of key chains). \qed
\end{proof}

\subsection{Performance Optimization}
\label{sec:improvement}

In this section we show how to reduce the communication complexity in the $\SSGC$;
this might be very useful in applications such as MANETs and sensor networks.
Suppose a {\sf Join} event occurs at time $t$. The key observation includes:
\begin{enumerate}
\item In $\join^*$ of $\SSGC$ we could simply let the server
sends the updated keys to the joining user $U$.
We notice that, before $U$ receiving the rekeying message from the server,
$K_{U}^{(t-1)}$ consists of a single key, denoted by $k^*$,
that is also known to the server. After sending the rekeying message,
the server update $K_{\GC}^{(t-1)}=\{k_i\}$ to $K_{\GC}^{(t)}=\{f_{k_i}(1)\} \cup \{f_{k^*}(1)\}$.

\item When the joining user $U$ executes $\rekey^*$
corresponding to the $\join^*$ (i.e., after receiving
the rekeying message),
it lets $f_{k^*}(0)$ plays the role of $k^*$.
Then, $U$ updates $k^*$ to $f_{k^*}(1)$ while keeping intact
the other keys received from the server.

\item When an existing user $V \in \Delta^{(t-1)}$
executes $\rekey^*$ corresponding to the $\join^*$, it
simply updates every $k_i \in K_V^{(t-1)}$ (including the group key) to $f_{k_i}(1)$.

\item The encryption of group communications is
based on new group key $k^{(t)}=f_{k^{(t-1)}}(1)$.
\end{enumerate}

We notice that the idea of substituting $k_i$ via a certain function was
pointed out in \cite{VersaKeyJSAC99,CanettiEurocrypt99} with respect to the specific scheme
of \cite{WGL00}. Here we show that it
can actually be extended to accommodate the class of group communication
schemes discussed in this paper. This justifies why we treat it as a possible
feature of the compiler, which we call the {\em optimized compiler}.

\begin{theorem}
\label{theorem:main2} Assume $\{f_k\}$ is a secure pseudorandom
function family, and $\SGC$ is {\tt secure}. If $\SGC$ does not
adopt the afore-discussed performance optimization (otherwise, the
optimized compiler does not gain anything over the original
compiler), then the scheme output by the optimized compiler is
also {\tt strongly-secure}.
\end{theorem}

The proof is similar to the proof of Theorem \ref{theorem:main1}, and thus omitted.

\section{A Concrete {\tt Strongly-Secure} Stateful Group Communication Scheme}
\label{sec:concrete}

In the last section we presented a compiler that can transform a certain {\tt secure} stateful
group communication scheme
into a {\tt strongly-secure} one. In this section we present a concrete {\tt strongly-secure}
stateful group communication scheme, which is obtained by applying the compiler to
LKH \cite{WGL00} that is shown to be {\tt secure} in Section \ref{sec:wgl-security}.
First we briefly review LKH.

\subsection{The Model of LKH}
\label{sec:a-model}

The model of LKH is best known as a {\em key tree}, which outperforms the others
(e.g., star key graph, or general key graph which actually leads to a certain NP-hard problem
as we always need to minimize the communication complexity).
A key tree $T$ can be seen as a special class of directed acyclic
graph with two types of nodes: {\em $u$-nodes}
representing users and {\em $k$-nodes} representing keys. Each $u$-node is a leaf
that has one outgoing edge but no incoming edge,
and each $k$-node is an inner node that has one or more incoming edges.
Moreover, there is a $k$-node (i.e., the root) that has incoming edges but no outgoing edge.
In other words, the edges go from leaves towards the root.

Let $U$ be a finite and nonempty set of users and $K$ be a finite and nonempty set of keys.
We are interested in a relation, $R \subseteq U \times K$, that can be
specified by a key tree $T$ as follows:
\begin{itemize}
\item There is a one-to-one correspondence between $U$ and the set of $u$-nodes in $T$.
\item There is a one-to-one correspondence between $K$ and the set of $k$-nodes in $T$.
\item $(u,k) \in R$ if and only if there is a directed path in $T$
from the $u$-node that corresponds to
a user $u \in U$ to the $k$-node that corresponds to a key $k \in K$.
\end{itemize}
This means that the group key is at the root of the tree,
which is shared by all the users in $U$.
Since a key tree can be specified by two parameters --
the height $h$ of the tree is the length (in number of edges) of the
longest directed path in the tree,
and the degree $d$ of the tree is the maximum number of incoming edges
of a node in the tree --
each user in $U$ has at most $h$ keys.

In order to clarify the presentation, we define two functions,
$\keyset: U \to K$ and $\userset: K \to U$, as follows:
\begin{eqnarray*}
\keyset(u) & = & \{k | (u,k) \in R\}, \\
\userset(k) & = & \{u | (u,k) \in R\}.
\end{eqnarray*}
Intuitively, $\keyset(u)$ is the set of keys held by user $u \in
U$, and $\userset(k)$ is the set of users that hold key $k \in K$.
Moreover, it is natural to generalize the definitions of
$\keyset(u \in U)$ to $\keyset(U' \subseteq U) =
\bigcup_{ u \in U'} \keyset(u)$, and of $\userset(k \in K)$ to
$\userset(K' \subseteq K)=\bigcup_{ k\in K'}\userset(k)$.

\subsection{A {\tt Strongly-Secure} Stateful Group Communication Scheme}
\label{sec:new-scheme}

The new scheme is obtained by applying the compiler described in Section \ref{sec:compiler}
to LKH based on the so-called {\em group-oriented}
rekeying strategy, which is reviewed in Appendix \ref{appendix:wgl} for completeness.
(LKH can be based on the less efficient {\em key-oriented} and
{\em user-oriented} strategies \cite{WGL00}. Nevertheless, it should be straightforward to
adapt our scheme to these rekeying strategies.) The scheme
consists of four protocols, namely $\SSGC=(\setup^*,\join^*,\leave^*,\rekey^*)$.

\smallskip

\noindent{\bf $\setup^*$:} The key server generates a key $k_i$ for each $k$-node. After the
initialization, each user (corresponding to a $u$-node) holds the
keys corresponding to the path from its parent
$k$-node to the root.

For example, if the initial system configuration is like in Figure \ref{fig:attack}.(a),
then user $u_5$ holds keys, $k_5, k_{456},k_{1-8}$, where $k_{1-8}$ is the group key.

\smallskip

\noindent{\bf $\join^*$:} After granting a join request from user
$u$, the key server $s$ creates a new $u$-node for user $u$ and a
new $k$-node for its individual key $k_u$. Then,
server $s$ finds an existing $k$-node (called the {\em joining
point} for this join request) in the key tree and attaches the
$k$-node $k_u$ to the joining point as its child. As a
consequence, the keys corresponding to the path -- starting at the
joining point and ending at the root -- need to be changed. The
algorithm is specified in Figure \ref{fig:join}, whose basic idea
can be summarized as follows:
\begin{enumerate}
\item For each $k$-node $x$ whose key needs to be changed, say, from $\bk_i$ to
freshly chosen $\hk_i$, the server constructs
two rekeying messages. The first rekeying message is the encryption of
new key $\hk_i$ with $f_{\bk_i}(0)$, where $\bk_i$ is a non-root key that needs to be changed, and is sent to
$\userset(\bk_i)$, namely the set of users that share $\bk_i$.
The second rekeying message contains the encryption of the
new key $\hk_i$ with the individual key
of the joining user, and is sent to the joining user. Moreover,
these rekeying messages are appropriately grouped together.
\item Any other key $\bk_j$ that needs not to be changed is replaced by $f_{\bk_j}(1)$.
\end{enumerate}

\begin{figure*}[ht]
\begin{center}
\fbox{
\begin{minipage}{41pc}
\begin{tabbing}
123\=123\=123\=\kill
Join protocol for group-oriented rekeying: \ \ \ \ \ \ \ \ // suppose user $u$ joins the group \\
\> server $s$ generates a new key $k_u$ for user $u$ \\
\> server $s$ finds a joining point $x_{j}$ \\
\> server $s$ attaches $k_u$ to $x_j$ \\
\> let $x_0$ be the root \\
\> denote by $x_{i-1}$ the parent of $x_i$ for $1 \leq i \leq j$ \\
\> $\bk_{j+1} \leftarrow k_u$ \\
\> let $\bk_0,\bk_1, \ldots,\bk_j$ be the current keys of $x_0,\ldots,x_j$, respectively \\
\> server $s$ generates fresh keys $\hk_0,\hk_1, \ldots, \hk_j$ \ \ \ \ //
new keys of $x_0,\ldots,x_j$ \\
\> $s \to \userset(\bk_0): \{\hk_0\}_{{\bk_0}}, \{\hk_1\}_{f_{\bk_1}(0)}, \ldots,\{\hk_j\}_{f_{\bk_j}(0)}$ \\
\> $s \to \{u\}: \{\hk_0, \hk_1, \ldots, \hk_j\}_{f_{\bk_{j+1}}(0)}$ \\
\> {\tt FOR} all $\bk \in (\keyset(\userset(\bk_0)) \setminus \{\bk_0, \bk_1, \ldots,\bk_j\})
\cup \{\bk_{j+1} \}\cup \{\hk_{1}, \ldots, \hk_{j} \}$ \\
\>\> $\bk \leftarrow f_{\bk}(1)$
\end{tabbing}
\end{minipage}
}
\end{center}
\caption{Join-incurred group-oriented rekeying}
\label{fig:join}
\end{figure*}

For example, if $u_9$ joins the group configured as in Figure \ref{fig:attack}.(a), then
$u_9$ is granted to join at joining point of $k$-node $k_{78}$.
Then, the group key is changed from $k_{1-8}$
to $k_{1-9}$, and $k_{78}$ is replaced with a new $k_{789}$.
The rekeying messages sent to the users are:
\begin{eqnarray*}
\begin{array}{lll}
s \to \{u_1, \ldots, u_8\} & : & \{k_{1-9}\}_{{k_{1-8}}}, \{k_{789}\}_{f_{k_{78}}(0)}, \\
s \to \{u_9 \}& : & \{k_{1-9}, k_{789} \}_{f_{k_9}(0)}.
\end{array}
\end{eqnarray*}
Finally, $k_i$ is substituted with $f_{k_i}(1)$ for $i \in \{1,2,3,4,5,6,7,8,9,123,456,789\}$.
The attack presented in the Introduction is blocked because, for example,
compromise of $f_{k_9}(1)$ does not lead to the exposure of $f_{k_9}(0)$,
where $k_9$ is not known to the adversary (because it has been securely erased).
As a result, the past group key
$k_{1-9}$ cannot be recovered by the adversary.

\smallskip

\noindent{\bf $\leave^*$:} After granting a leave request from
user $u$, the key server $s$ updates the key tree by deleting the
$u$-node for user $u$ and the $k$-node for its individual key from
the key tree. The parent of the $k$-node corresponding to the
user's individual key is called the {\em leaving point}. As a
consequence, the keys corresponding to the path -- starting at the
leaving point and ending at the root -- need to be changed.
The algorithm is specified in Figure \ref{fig:leave},
whose basic idea can be summarized as follows:
\begin{enumerate}
\item For each $k$-node $x$ whose key needs to be changed, say, from $\bk_i$ to
freshly chosen $\hk_i$, the server constructs
a rekeying message that is the encryption of
$\hk_i$ with the keys of $x$'s children in the new key tree.
Note that ``the keys of $x$'s children in the new key tree" are either certain new keys
that need to be distributed, or some current keys that need not to be changed (although
they will be appropriately updated).
\item Any other key $\bk_j$ that needs not to be changed is replaced by $f_{\bk_j}(1)$.
\end{enumerate}

\begin{figure*}[ht]
\begin{center}
\fbox{
\begin{minipage}{41pc} 
\begin{tabbing}
123\=123\=123\=\kill
Leave protocol for group-oriented rekeying: \ \ \ \ \ \ \ \ // suppose $u$ leaves the group \\
\> let $x_{j+1}$ be the deleted $k$-node for $k_u$ \\
\> $\bk_{j+1} \leftarrow k_u$ \\
\> server $s$ finds the leaving point $x_j$ (parent of $k_u$)  \\
\> server $s$ removes $\bk_{j+1}$ from the key tree \\
\> let $x_0$ be the root \\
\> denote by $x_{i-1}$ the parent of $x_i$ where $1 \leq i \leq j$ \\
\> let $\bk_0,\bk_1, \ldots, \bk_j$ be the keys of $x_0, x_1, \ldots,x_j$ \ \ \ // they need to be changed \\
\> server $s$ generates fresh keys $\hk_0,\hk_1, \ldots, \hk_j$ as the new keys of $x_0, x_1, \ldots,x_j$ \\
\> {\tt FOR} $i = 0 \text{ \tt TO } j-1$ \\
\>\> let $\bk_{i_1}, \ldots, \bk_{i_{z_i}}$ be the keys at the children of $x_i$ in the new key tree \\
\>\>\> where $\bk_{i_a}$ is to be changed to $\hk_{i+1}$ for some $a \in \{1, \ldots, {z_i}\}$ \\
\>\> $L_i \leftarrow (\{\hk_i\}_{f_{\bk_{i_1}}(0)}, \ldots, \{\hk_i\}_{f_{\bk_{i_{a-1}} (0)}},
\{\hk_i\}_{f_{\hk_{i+1}}(0)}, \{\hk_i\}_{f_{\bk_{i_{a+1}}}(0)}, \ldots, \{\hk_i\}_{f_{\bk_{i_z}}(0)})$ \\
\> let $\bk_{j_1}, \ldots, \bk_{j_{z_j}}$ be the keys at the children of $x_j$ in the new key tree \\
\> $L_j \leftarrow (\{\hk_j\}_{f_{\bk_{j_1}}(0)}, \ldots, \{\hk_j\}_{f_{\bk_{j_z}}(0)})$ \\
\> $s \to \userset(\bk_0) \setminus \{u\}: (L_0, \ldots, L_j)$ \\
\> {\tt FOR} all $\bk \in (\keyset(\userset(\bk_0)) \setminus \{\bk_0,\bk_1,\ldots,\bk_{j+1}\}) \cup
\{\hk_{1}, \ldots, \hk_{j}\}$ \\
\>\> $\bk \leftarrow f_{\bk}(1)$
\end{tabbing}
\end{minipage}
}
\end{center}
\caption{Leave-incurred group-oriented rekeying}
\label{fig:leave}
\end{figure*}

For example, if $u_8$ leave the group as configured in Figure \ref{fig:attack}.(b),
the leaving point is the $k$-node $k_{789}$. Then, the group key is changed from $k_{1-9}$
to $k_{1-7,9}$, and the key of leaving point is changed from $k_{789}$ to
$k_{7,9}$ in Figure \ref{fig:attack}.(c).
The rekeying message sent to the users is:
\begin{eqnarray*}
s \to \{u_1, \ldots, u_7, u_9\}: \{k_{1-7,9}\}_{f_{k_{123}}(0)}, \{k_{1-7,9}\}_{f_{k_{456}}(0)},
\{k_{1-7,9}\}_{f_{k_{7,9}}(0)}, \{k_{7,9}\}_{f_{k_7}(0)}, \{k_{7,9}\}_{f_{k_9}(0)}.
\end{eqnarray*}
Finally, $k_i$ is updated to $f_{k_i}(1)$ for
$i \in \{1,2,3,4,5,6,7,9,123,456\}$, and $k_{7,9}$ is updated to $f_{k_{7,9}}(1)$.
The attack presented in the Introduction is blocked because, for example,
compromise of $f_{k_{9}}(1)$ does not lead to the exposure of $f_{k_{9}}(0)$,
where $k_{9}$ is not known to the adversary (because it has been securely erased).
As a result, $k_{7,9}$, and thus the past group key
$k_{1-7,9}$ cannot be recovered by the adversary.

\smallskip

\noindent{\bf $\rekey^*$:} If the rekeying event is incurred by a
join event, a legitimate user (i.e., an existing one or a joining
one) obtains a subset $\Theta'$ of $\Theta=\{\hk_0, \hk_1, \ldots,
\hk_{j}\}$, and updates each $\hk \in \Theta' \setminus \{\hk_0\}$ to $f_{\hk}(1)$.
If the rekeying event is incurred by a leave event,
a legitimate user (i.e., one remaining in the group) obtains a subset $\Theta'$ of
$\Theta=\{\hk_0, \hk_1, \ldots, \hk_{j}\}$,
and updates each $\hk \in \Theta' \setminus \{\hk_0\}$ to $f_{\hk}(1)$.
In any case, a legitimate user $u$ updates each $k_i \in \keyset(u)$
to $f_{k_i}(1)$, as long as $k_i$ is not changed to any key belonging to $\Theta$,
and erases the outdated keys.

For example, corresponding to the event that $u_9$ joins the group
as shown in Figure \ref{fig:attack}.(a), $u_1$ obtains $k_{1-9}$,
updates $k_{123}$ to $f_{k_{123}}(1)$, and updates $k_1$ to
$f_{k_1}(1)$. Whereas, $u_9$ obtains $k_{1-9}$ as well as
$k_{789}$, updates $k_{789}$ to $f_{k_{789}}(1)$, and updates
$k_9$ to $f_{k_9}(1)$. Corresponding to the event that $u_8$
leaves the group as shown in Figure \ref{fig:attack}.(b), $u_1$
obtains $k_{1-7,9}$, updates $k_{123}$ to $f_{k_{123}}(1)$, and
updates $k_1$ to $f_{k_1}(1)$. Whereas, $u_9$ obtains $k_{1-7,9}$
as well as $k_{7,9}$, updates $k_{7,9}$ to $f_{k_{7,9}}(1)$, and
updates $k_9$ to $f_{k_9}(1)$.

\subsection{Analysis}
\label{sec:wgl-security}

\begin{theorem}
\label{theorem:wgl} Assume that the stand-alone encryptions
utilized in LKH are based on a secure
pseudorandom function family. Then, LKH is {\tt secure}.
\end{theorem}

\begin{proof} (sketch)
Consider a mental scheme that is the same as LKH,
except that the encryptions corresponding to the rekeying event of
interest -- the event corresponding to the
$\OO_{Test}(\cdot,\cdot)$ query -- are based on random functions.
This substitution can get through because the definition of {\tt security}
requires that there are no corrupt users.
We claim that this mental scheme is {\tt secure}; otherwise, a
standard hybrid argument shows that the pseudorandom function family is
broken because the number of encryptions is polynomially-bounded
(which is further based on the fact that the size of the key-tree
is polynomially-bounded). Conditioned on the fact that there are a
polynomially-bounded number of rekeying events, we conclude that
LKH is {\tt secure}. \qed
\end{proof}

As a corollary of Theorem \ref{theorem:main1} (which states that
the compiler transforms a {\tt secure} stateful group
communication scheme to a {\tt strongly-secure} one) and Theorem
\ref{theorem:wgl} (which states that LKH is indeed {\tt
secure}), we have:

\begin{corollary}
\label{corollary:main}
The scheme presented in Section \ref{sec:new-scheme}
is {\tt strongly-secure}.
\end{corollary}

\subsection{Performance Optimization}
\label{sec:stateful-optimized-scheme}

The improved scheme differs from $\SSGC$ only in $\join^*$ and $\rekey^*$.

\smallskip

\noindent{\bf Improved $\join^{*}$:} This algorithm is specified in Figure \ref{fig:join1}.
It is the same as the $\join^*$ except the following: (1)
instead of freshly choosing new keys for the $k$-nodes on the
path starting at a joining point and ending at the root,
we simply update every existing key $k$ as $f_k(1)$, and (2) the new group key for encrypting actual
group communications is $f_k(0)$, where $k$ is the already updated key at the root.

\begin{figure*}[ht]
\begin{center}
\fbox{
\begin{minipage}{41pc} 
\begin{tabbing}
123\=123\=123\=\kill
Join protocol for group-oriented rekeying: \ \ \ \ \ \ \ \  // suppose user $u$ joins the group \\
\> server $s$ generates a new key $k_u$ for user $u$  \\
\> server $s$ finds a joining point $x_{j}$ \\
\> server $s$ attaches $k_u$ to $x_j$ \\
\> let $x_0$ be the root \\
\> denote by $x_{i-1}$ the parent of $x_i$ for $1 \leq i \leq j$ \\
\> let $\bk_0,\bk_1, \ldots,\bk_j$ be the current keys of $x_0,\ldots,x_j$, respectively \\
\> $\bk_{j+1} \leftarrow k_u$  \\
\> $s \to \userset(\bk_0):$ ``key update" \\
\> $s \to \{u\}: \{f_{\bk_0}(1), f_{\bk_1}(1), \ldots, f_{\bk_j}(1)\}_{f_{\bk_{j+1}}(0)}$ \\
\> {\tt FOR} all $\bk \in \keyset(\userset(\bk_0)) \cup \{\bk_{j+1}\}$ \\
\>\> $\bk \leftarrow f_{\bk}(1)$
\end{tabbing}
\end{minipage}
}
\end{center}
\caption{Improved join-incurred group-oriented rekeying}
\label{fig:join1}
\end{figure*}

\noindent{\bf Improved $\rekey^*$:} It is the same as $\rekey^*$ except that when
the rekeying is incurred by a join event:
every existing user $v$ holding a key set $\keyset(v)$
needs to update every $k \in \keyset(v)$ to $f_{k}(1)$,
whereas the joining user $u$
only needs to update its common key $k_u$, which is established
during the out-of-band approval of the join request, to $f_{k_u}(1)$.
Note that the new group key for encrypting actual
group communications is $f_k(0)$, where $k$ is the already updated key at the root.

\smallskip

For example, if $u_9$ joins the group configured as in Figure \ref{fig:attack}.(a), then
$u_9$ is granted to join at the joining point of $k$-node $k_{78}$.
The key at the root is updated from $k_{1-8}$
to $k_{1-9} =f_{k_{1-8}}(1)$ such that the new key for encrypting actual group communications
is $f_{k_{1-9}}(0)$, and $k_{78}$ is updated to $k_{789} =f_{k_{78}}(1)$.
The sever sends the following messages:
\begin{eqnarray*}
\begin{array}{lll}
s \to \{u_1, \ldots, u_8\} & : & \text{``key update"} \\
s \to \{u_9\} & : & \{k_{1-9}, k_{789} \}_{f_{k_9}(0)}.
\end{array}
\end{eqnarray*}
Every existing user $u_i$, $i \in \{1,2,3,4,5,6,7,8\}$, with key set $K_i$,
updates every $k \in K_i$ as $f_{k}(1)$.
For example, $u_7$ updates $k_{1-8}$ to $k_{1-9} =f_{k_{1-8}}(1)$,
updates $k_{78}$ to $k_{789} =f_{k_{78}}(1)$, and
updates $k_7$ to $k_7 = f_{k_7}(1)$.
On the other hand, the joining user $u_9$ only
needs to update $k_9$ to $k_9=f_{k_9}(1)$, which means that it keeps
$(k_{1-9}, k_{89}, k_9 = f_{k_9}(1))$.

As a corollary of Theorem \ref{theorem:main2} (which states that the
optimized compiler in Section \ref{sec:improvement} transforms a {\tt secure}
stateful group communication scheme into a {\tt strongly-secure} one)
and Theorem \ref{theorem:wgl} (which states that LKH is indeed {\tt secure}), we have

\begin{corollary}
The optimized scheme in Section \ref{sec:stateful-optimized-scheme}
is {\tt strongly-secure}.
\end{corollary}

\section{The Case of Stateless Group Communication Schemes}
\label{sec:stateless-case}

Recall that we briefly reviewed the subset-cover
framework \cite{NaorCrypto01} in Section \ref{sec:motivation}.
This section is organized as follows. In Section
\ref{sec:stateless:security-definition} we discuss the models
and security definitions, including the
notions of {\tt strong-security} (i.e., backward-security in the
active outsider attack model) and of {\tt security} (i.e.,
forward-security in the passive attack model).
In Section \ref{sec:stateless:relationships}
we explore the relationships between the security notions. In Section
\ref{sec:stateless:compiler} we present a compiler that can
transform a {\em subclass} of {\tt secure} stateless group
communication schemes into {\tt strongly-secure} ones, whose
security is analyzed in Section
\ref{sec:stateless:security-analysis}. A concrete {\tt
strongly-secure} stateless group communication scheme, which is
based on the {\em complete subtree method} \cite{NaorCrypto01}, is
presented in Section \ref{sec:stateless:concrete-scheme}.
Some practical issues are
discussed in Section \ref{sec:stateless:discussion}.

\subsection{Model and Security of Stateless Group Communication Schemes}
\label{sec:stateless:security-definition}

The subset-cover framework of \cite{NaorCrypto01} was briefly reviewed in
Fig. \ref{fig:stateless:framework}. More specifically, let $\kappa$ be a security
parameter, ${\cal N}$ be the set of all users such that $|{\cal
N}| =N$ is polynomially-bounded, and ${\cal R} \subset {\cal N}$
be a group of $|{\cal R}|=r$ users whose decryption privileges
should be revoked. Let $E_L$ be a symmetric key cryptosystem
secure against an adaptive chosen-plaintext attack, and $F_K$ be a
symmetric key cryptosystem with a weaker security property called
indistinguishability under a single chosen-plaintext
attack in \cite{NaorCrypto01} (which is called ``IND-P0-C0 security"
in \cite{KYSTOC00}).\footnote{Notice that \cite{NaorCrypto01}
required that $E_L$ be secure against chosen-ciphertext attacks, whereas
we require it to be secure against chosen-plaintext attacks.
The reason is that we need to we assume that the underlying communication channels
are authenticated. While this naturally prevents chosen-ciphertext attacks,
it also avoid another subtle attack, namely that a dishonest user could
successfully cheat an honest user into accepting an impersonating message.
The reason is simply due to the fact that $E_L$ being secure against
chosen-ciphertext attacks does not necessarily prevent this attack,
because the dishonest user also knows the common secret key.
This subtlety is well understood in the context of group communications
(cf. \cite{BonehEurocrypt01}).
}

Recall that the goal of a stateless group communication scheme is to allow a center
(or group controller, server, or sender) to transmit a message $M$
to all users such that any user $u \in {\cal N} \setminus {\cal R}$ can decrypt
the message correctly, while even a coalition consisting of all members of ${\cal R}$
cannot decrypt it.
Suppose $S_1,\ldots,S_w$ are a collection of subsets of users, where
$S_j \subseteq {\cal N}$ for $1 \leq j \leq w$, and each $S_j$ is assigned a long-lived
key $L_j$ such that each $u \in S_j$ should be able to
deduce $L_j$ from its secret information $I_u$.
Given a revoked set ${\cal R}$, if one can partition ${\cal N} \setminus {\cal R}$
into (ideally disjoint) sets $S_{i_1},\ldots,S_{i_m}$ such that
${\cal N} \setminus {\cal R} \subseteq \cup_{\ell=1}^m S_{i_\ell}$, then
a message-encryption key $K$ can be encrypted $m$ times with $L_{i_1},\ldots,L_{i_m}$, and
each user $u \in {\cal N} \setminus {\cal R}$ can obtain $K$ and thus $M$.

In what follows, by ``$\A$ corrupts a user $u$" we mean that
not only the internal state of $u$ (including $I_u$) is given to $\A$, but also
$u$ will behave under $\A$'s control (i.e., Byzantine);
by ``$u$ is revoked" we mean that $u$ is not entitled to receive the message with respect to
the specified session(s).
For simplicity, we assume that a user, once corrupted, is always corrupt.

Stateless group communication schemes are indeed simpler than
stateful ones because (1) both the joining and leaving operations
are implicit --- the rekeying messages may even be coupled with the
payload, and (2) when a user (re-)joins a group, its long-term
keys can indeed be reused. Therefore, the model of stateless group
communication schemes can also be correspondingly simplified. In
particular, we assume the center keeps an incremental counter for
each broadcast messages so that encryptions
may be simply denoted by $C^{(1)},C^{(2)},\ldots$ and the
corresponding plaintext messages may be denoted by
$M^{(1)},M^{(2)},\ldots$. One may think each $C^{(i)}$ corresponds
to an ``rekeying" event with revocation set ${\cal R}^{(i)}$ for
$i=1,2,\ldots$. Note that ${\cal R}^{(i)} \neq {\cal R}^{(i+1)}$.

We assume that during the system initialization the center can
communicate with each legitimate user through an {\em authenticated
private} channel. In practice, the authenticated private channel can
be implemented by a two-party authenticated key-exchange protocol,
which should also ensure, as in the case of stateful group
communication schemes, that certain relevant keys are securely
erased after the initialization. Further, we assume that after the
system initialization the center can communicate with a user
through an {\em authenticated} channel.

In parallel to the case of stateful group communication schemes,
we define two adversarial models for stateless group communication schemes:
the {\em active outsider attack model} and the {\em passive outsider attack model}.

\begin{definition}\emph{(active outsider attack model)}
\label{def:stateless:adversarial-model-active}
With respect to a given $i^{th}$ broadcast message $C^{(i)}$,
we say a user $u$ is legitimate if $u \in {\cal N} \setminus {\cal R}^{(i)}$, and
is illegitimate (or an outsider) otherwise.
By ``active outsider attack model" we mean the adversarial model in which
an outsider $\A$ of the $i^{th}$ broadcast message, which
may be called the ``challenge" message,
is allowed to corrupt legitimate users of $C^{(j)}$ for any $j>i$ (i.e.,
$u \in {\cal N} \setminus {\cal R}^{(j)}$).
\end{definition}

\begin{definition}\emph{(passive attack model)}
\label{def:stateless:adversarial-model-passive}
In this adversarial model, the adversary $\A$ is not allowed to corrupt any
other legitimate member. In other words, the adversary is only allowed to
decide when it is to be revoked (though in an arbitrary fashion).
Formally, $\A$ cannot corrupt any
$u \in {\cal N} \setminus \{\A\}$ for $i=1,2,\ldots$.
\end{definition}

In each of the two models, we define two security notions: backward-security and forward-security.
That is, we have four security notions:
(1) forward-security in the active outsider attack model or simply {\tt security} for short,
(2) backward-security in the active outsider attack model or simply {\tt strong-security} for short,
(3) forward-security in the passive attack model, and
(4) backward-security in the passive attack model.

\begin{definition}\emph{({\tt security}; adapted from \cite{NaorCrypto01})}
\label{def:stateless:security}
Consider an adversary $\A$ that gets to
\begin{enumerate}
\item Select adaptively ${\cal R}^{(1)}, {\cal R}^{(2)}, \ldots,{\cal R}^{(\ell_1)}$ of receivers,
obtain $I_u$ for all $u \in {\cal R}^{(i)}$ and see $C^{(1)}, C^{(2)}, \ldots,C^{(\ell_1)}$
for $i=1,2,\ldots,\ell_1$.

\item Choose a message $M$ as the challenge plaintext and a set ${\cal R}$
of revoked users that must include
all the ones it corrupted (but may contain more); i.e.,
$\cup_{i=1}^{\ell_1}{\cal R}^{(i)} \subseteq {\cal R}$.
$\A$ then receives an encrypted message $C$ with a revoked set ${\cal R}$, where $C$ is
the encryption of either $M$ or a random message of the same length. We may call this
the ``challenge" message.

\item For $i=\ell_1+2,\ell_1+3,\ldots$,
the following restrictions apply.
(1) Even if $\A \in {\cal N}$, $\A$ can only decide whether $\A \in {\cal R}^{(i)}$.
(2) Even if ${\cal R}^{(i)} \setminus \{\A\} \neq \emptyset$, $\A$ has no access to
any $I_u$ for $u \in {\cal R}^{(i)} \setminus \{\A\}$.
\end{enumerate}
Now $\A$ has to guess whether $C$ corresponds to the encryption of the real message $M$ or a random message.
Denote by {\sf Succ} the event that $\A$ makes the right guess.
The advantage of $\A$ is defined as
${\sf Adv}_{\A}(\kappa) = | 2 \cdot \Pr[{\sf Succ}] -1|$,
where $\Pr[{\sf Succ}]$ is the probability that the event {\sf Succ} occurs, and the probability
is taken over the coins used by the center and by $\A$.
We say that a stateless group communication scheme is {\tt secure}
(or forward-secure in the active outsider attack model) if, for any probabilistic
polynomial-time $\A$ as above, it holds that
${\sf Adv}_{\A}(\kappa)$ is negligible in $\kappa$.
\end{definition}

\begin{definition}\emph{({\tt strong-security})}
\label{def:stateless:strong-security}
Consider an adversary $\A$ that gets to
\begin{enumerate}
\item Select adaptively ${\cal R}^{(1)}, {\cal R}^{(2)}, \ldots,{\cal R}^{(\ell_1)}$ of receivers,
obtain $I_u$ for all $u \in {\cal R}^{(i)}$ and see $C^{(1)}, C^{(2)}, \ldots,C^{(\ell_1)}$
for $i=1,2,\ldots,\ell_1$.

\item Choose a message $M$ as the challenge plaintext and a set ${\cal R}$
of revoked users that must include
all the ones it corrupted (but may contain more); i.e.,
$\cup_{i=1}^{\ell_1}{\cal R}^{(i)} \subseteq {\cal R}$.
$\A$ then receives an encrypted message $C$ with a revoked set ${\cal R}$, where $C$ is
the encryption of either $M$ or a random message of the same length. We may call this
the ``challenge" message.

\item Select adaptively ${\cal R}^{(\ell_1+2)}, {\cal R}^{(\ell_1+3)}, \ldots$ of receivers
and obtain $I_u$ for all $u \in {\cal R}^{(i)}$ for $i=\ell_1+2, \ell_1+3, \ldots$.
Besides, $\A$ may select messages $M^{(\ell_1 + 2)},M^{(\ell_1 + 3)},\ldots$
and see the encryption of $C^{(\ell_1 + 2)},C^{(\ell_1 + 3)},\ldots$.
\end{enumerate}
Now $\A$ has to guess whether $C$ corresponds to the
encryption of the real message $M$ or a random message.
Denote by {\sf Succ} the event that $\A$ makes the right guess.
The advantage of $\A$ is defined as
${\sf Adv}_{\A}(\kappa) = | 2 \cdot \Pr[{\sf Succ}] -1|$,
where $\Pr[{\sf Succ}]$ is the probability that the event {\sf Succ} occurs, and the probability
is taken over the coins used by the center and by $\A$.
We say that a stateless group communication scheme is {\tt strongly-secure}
(or backward-secure in the active outsider attack model) if, for any probabilistic
polynomial-time $\A$ as above, it holds that
${\sf Adv}_{\A}(\kappa)$ is negligible in $\kappa$.
\end{definition}

\begin{definition}\emph{(forward-security in the passive attack model)}
\label{def:stateless:forward-secrecy-passive}
Consider an adversary $\A$ that gets to
\begin{enumerate}
\item Select adaptively ${\cal R}^{(1)}, {\cal R}^{(2)}, \ldots,{\cal R}^{(\ell_1)}$ of receivers,
and see $C^{(1)}, C^{(2)}, \ldots,C^{(\ell_1)}$ for $i=1,2,\ldots,\ell_1$.
However, $\A$ does not have access to any $I_u$, where $u \in {\cal R}^{(i)}\setminus \{\A\}$
and $i=1,2,\ldots,\ell_1$.

\item Choose a message $M$ as the challenge plaintext and a set ${\cal R}$
of revoked users that must include
all the ones it corrupted (but may contain more); i.e.,
$\cup_{i=1}^{\ell_1}{\cal R}^{(i)} \subseteq {\cal R}$.
$\A$ then receives an encrypted message $C$ with a revoked set ${\cal R}$, where $C$ is
the encryption of either $M$ or a random message of the same length. We may call this
the ``challenge" message.

\end{enumerate}
Now $\A$ has to guess whether $C$ corresponds to the encryption of the real message $M$ or a random message.
Denote by {\sf Succ} the event that $\A$ makes the right guess.
The advantage of $\A$ is defined as
${\sf Adv}_{\A}(\kappa) = | 2 \cdot \Pr[{\sf Succ}] -1|$,
where $\Pr[{\sf Succ}]$ is the probability that the event {\sf Succ} occurs, and the probability
is taken over the coins used by the center and by $\A$.
We say that a stateless group communication scheme is {\tt secure}
(or forward-secure in the active outsider attack model) if, for any probabilistic
polynomial-time $\A$ as above, it holds that
${\sf Adv}_{\A}(\kappa)$ is negligible in $\kappa$.
\end{definition}

\begin{definition}\emph{(backward-security in the passive attack model)}
\label{def:stateless:backward-secrecy-passive}
Consider an adversary $\A$ that gets to
\begin{enumerate}
\item Select adaptively ${\cal R}^{(1)}, {\cal R}^{(2)}, \ldots,{\cal R}^{(\ell_1)}$ of receivers,
and see $C^{(1)}, C^{(2)}, \ldots,C^{(\ell_1)}$ for $i=1,2,\ldots,\ell_1$.
However, $\A$ does not have access to any $I_u$, where $u \in {\cal R}^{(i)} \setminus \{\A\}$ and
$i=1,2,\ldots,\ell_1$.

\item Choose a message $M$ as the challenge plaintext and a set ${\cal R}$
of revoked users that must include
all the ones it corrupted (but may contain more); i.e.,
$\cup_{i=1}^{\ell_1}{\cal R}^{(i)} \subseteq {\cal R}$.
$\A$ then receives an encrypted message $C$ with a revoked set ${\cal R}$, where $C$ is
the encryption of either $M$ or a random message of the same length. We may call this
the ``challenge" message.

\item Select adaptively ${\cal R}^{(\ell_1+2)}, {\cal R}^{(\ell_1+3)}, \ldots$ of receivers,
and possibly select messages $M^{(\ell_1 + 2)},M^{(\ell_1 + 3)},\ldots$
and see the encryption of $C^{(\ell_1 + 2)},C^{(\ell_1 + 3)},\ldots$.
However, $\A$ does not have access to any $I_u$, where $u \in {\cal R}^{(i)} \setminus \{\A\}$
and $i=\ell_1+2, \ell_1+3, \ldots$.
\end{enumerate}
Now $\A$ has to guess whether $C$ corresponds to the
encryption of the real message $M$ or a random message.
Denote by {\sf Succ} the event that $\A$ makes the right guess.
The advantage of $\A$ is defined as
${\sf Adv}_{\A}(\kappa) = | 2 \cdot \Pr[{\sf Succ}] -1|$,
where $\Pr[{\sf Succ}]$ is the probability that the event {\sf Succ} occurs, and the probability
is taken over the coins used by the center and by $\A$.
We say that a stateless group communication scheme is {\tt strongly-secure}
(or backward-secure in the active outsider attack model) if, for any probabilistic
polynomial-time $\A$ as above, it holds that
${\sf Adv}_{\A}(\kappa)$ is negligible in $\kappa$.
\end{definition}

It is trivial to see that {\tt strong-security} implies
{\em backward-security in the passive model}, and that
{\tt security} implies {\em forward-secure in the passive attack model}.

\subsection{Relationships between the Security Notions}
\label{sec:stateless:relationships}

We summarize the relationships between the security notions of
stateless group communication schemes in Fig.
\ref{fig:stateless:relationships}, where $X \rightarrow Y$ means
$X$ is stronger than $Y$, $X \leftrightarrow Y$ means $X$ is
equivalent to $Y$, $X \not\rightarrow Y$ means $X$ does not imply
$Y$, and $X \stackrel{?}{\rightarrow} Y$ means it is unclear where
$X$ implies $Y$. Below we elaborate on the non-trivial
relationships showed in Fig. \ref{fig:stateless:relationships}.

\begin{figure*}[h]
\unitlength=1.00mm \special{em:linewidth 0.4pt} \linethickness{0.4pt}
\begin{picture}(132.00,40.00)
\put(16.00,38.00){{backward-security in the}}
\put(16.00,34.00){{active outsider attack model}}
\put(16.00,30.00){{(i.e., {\tt strong-security})}}
\put(15.00,29.00){\framebox(46, 12)}
\put(68.00,40.00){{Proposition \ref{proposition:stateless:active-model:backward-implies-forward}}}
\put(61.00,39.00){\vector(1,0){38}}
\put(68.00,35.00){{Proposition \ref{proposition:stateless:active-model:backward-does-not-imply-forward}}}
\put(78.00,32.00){/}
\put(99.00,33.00){\vector(-1,0){38}}
\put(16.00,8.00){{backward-security in the}}
\put(16.00,2.00){{passive attack model}}
\put(15.00,0.00){\framebox(46, 12)}
\put(14.00,20.00){trivial}
\put(24.00,29.00){\vector(0,-1){17.00}}
\put(54.00,20.00){Theorem \ref{theorem:stateless:counter-example}}
\put(52.00,21.00){/}
\put(53.00,12.00){\vector(0,1){17.00}}
\put(100.00,38.00){{forward-security in the}}
\put(100.00,34.00){{active outsider attack model}}
\put(100.00,30.00){{(i.e., {\tt security})}}
\put(99.00,29.00){\framebox(46, 12)}
\put(68.00,7.00){{Proposition \ref{proposition:stateless:passive-model:backward-equals-to-forward}}}
\put(81.00,6.00){\vector(1,0){18}}
\put(81.00,6.00){\vector(-1,0){20.00}}
\put(100.00,8.00){{forward-security in the}}
\put(100.00,2.00){{passive attack model}}
\put(99.00,0.00){\framebox(46, 12)}
\put(98.00,20.00){trivial}
\put(108.00,29.00){\vector(0,-1){17.00}}
\put(136.00,21.00){{/}}
\put(139.00,18.00){{\huge ?}}
\put(137.00,12.00){\vector(0,1){17.00}}
\end{picture}
\caption{The relationships between the security notions in stateless group communication schemes}
\label{fig:stateless:relationships}
\end{figure*}
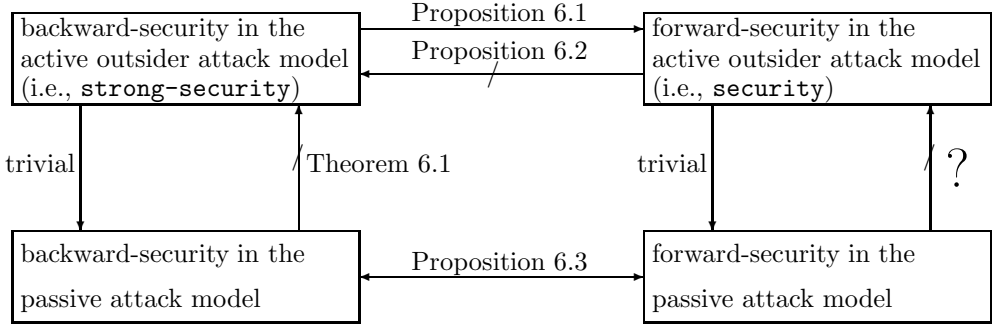

\begin{proposition}
\label{proposition:stateless:active-model:backward-implies-forward}
If a stateless group communication scheme is {\tt strongly-secure}, then it is also {\tt secure}.
\end{proposition}

\begin{proof}
This is almost immediate because, on one hand, the definition of {\tt strong-security}
ensures the secrecy of the encrypted content of $C$
even if $\A$ can have access to $I_u$ for $u \in {\cal N} \setminus {\cal R}^{(i)}$
for $i = \ell_1+2,\ell_1+3,\ldots$,
and on the other hand,
the definition of {\tt security}
ensures the secrecy of the encrypted content of $C$ only if $\A$ does not
have access to any $I_u$ such that $u \in {\cal N} \setminus {\cal R}^{(i)}$ and
$i \in \{\ell_1+2,\ell_1+3,\ldots\}$.
\end{proof}

\begin{proposition}
\label{proposition:stateless:active-model:backward-does-not-imply-forward}
A stateless group communication scheme that is {\tt secure} is not
necessarily {\tt strong-secure}.
\end{proposition}

\begin{proof}
The fact that {\tt security} does not imply {\tt strong-security} is implied by (1) Theorem
\ref{theorem:stateless:security}, which states that the complete subtree method of
the subset-cover framework is {\tt secure}, and (2) that
the subset-cover framework is insecure against an active outsider
attacker (cf. the attack scenario in
Section \ref{sec:motivation}). The key observation is indeed that
the adversary's capability in the {\tt strong-security} is strictly stronger.
\qed
\end{proof}

The above proposition implies that for a stateless group communication scheme, one only needs to
show that it is {\tt strongly-secure}.

\begin{proposition}
\label{proposition:stateless:passive-model:backward-equals-to-forward}
A stateless group communication scheme is backward-secure in the passive attack model
iff it is forward-secure in the passive attack model.
\end{proposition}

\begin{proof}
First we show that a stateless group communication scheme
that is not forward-secure in the passive attack
model is also not backward-secure in the passive attack model.
Suppose $\A$ is legitimate with respect to $C^{(i_1)}$, and illegitimate
with respect to $C^{(i_2)}$ where $i_1 < i_2$.
Since the scheme is not forward-secure in the passive attack
model, $\A$ can derive some information about $M$ corresponding to the $({\ell_1+1})^{th}$ broadcast $C$
with a non-negligible probability, where $i_2 \leq \ell_1+1$.
Now suppose $\A$ is legitimate with respect to $C^{(i_3)}$ where $\ell_1+1 < i_3$.
Then, with respect to $C^{(i_3)}$, $\A$
can derive some information about
a past encrypted message $M$ with respect to the $(\ell_1+1)^{th}$ broadcast
with a non-negligible probability.
Since $\A$ does not corrupt any other legitimate user $u \in {\cal R}^{j} \setminus \{\A\}$
for $j=\ell_1+2, \ell_1+3,\ldots$,
the scheme is not backward-secure in the passive attack model.

Second we show that a group communication scheme that is not
backward-secure in the passive attack model is also not
forward-secure in the passive attack model. Suppose $\A \in {\cal
N} \setminus {\cal R}_{i_1}$, $\A \notin {\cal N} \setminus {\cal
R}_{i_2}$, and $\A \in {\cal N} \setminus {\cal R}_{i_3}$, where
$i_1 < i_2 < i_3 $. Since the scheme is not backward-secure in the
passive attack model, without loss of generality, $\A$ can derive
some information about $M^{(i)}$ for some $i_2 \leq i < i_3$ with a non-negligible
probability. This also means that, with respect to $C^{(i_{2})}$,
$\A$ can derive some information about a future message $M^{(i)}$ for some $i \geq i_2$.
Since $\A$ does not corrupt any other legitimate users, the scheme
is not forward-secure in the passive attack model. \qed
\end{proof}

We do not know whether forward-security in the passive attack mode
also implies forward-security in the active outsider attack model.
The relationship may seem trivial at a first glance, since all the
corrupt members are revoked before the ``challenge" session, and
the adversary is not allowed to corrupt any member after the
``challenge" session. Although it can indeed be shown that the
implication holds, provided that the adversary is {\em static}
(meaning that the adversary decides which principals in ${\cal N}$
it will corrupt), in the more interesting case that the adversary
is {\em adaptive}, we do not know how to prove it.

\begin{theorem}
\label{theorem:stateless:counter-example}
There exists a stateless group communication scheme that
is ``backward-secure in the passive attack model"
but not {\tt strongly-secure} (i.e., backward-secure in the active outsider attack model).
\end{theorem}

\begin{proof}
Theorem \ref{theorem:stateless:security} shows that the complete subtree revocation scheme in
the subset-cover framework is {\tt secure}
(i.e., forward-secure in the active outsider attack model), which trivially means that
it is also forward-secure in the passive attack model.
Then, Proposition \ref{proposition:stateless:passive-model:backward-equals-to-forward}
shows that it is also backward-secure in the passive attack model.

On the other hand, the attack scenario showed in Section
\ref{sec:motivation} states that the subset-cover framework
is not backward-secure in the active outsider attack model. \qed
\end{proof}

\subsection{A Compiler for Stateless Group Communication Schemes}
\label{sec:stateless:compiler}

Now we present a compiler that can transform a subclass
of {\tt secure} stateless group communication schemes
falling into the subset-cover framework (called the input schemes)
into {\tt strongly-secure} ones.
The subclass of stateless group communication schemes has the characteristics that
the different keys belonging to $\{L_i\}_i \cup \{I_u\}_u$ are computationally
independent of each other. Let $\{f_k\}$ be a pseudorandom function family.
The compiler is specified in Fig. \ref{fig:stateless:strong-security}.

\begin{figure*}[h]
\begin{center}
\fbox{
\begin{minipage}{38pc} 
\begin{description}
\item[Initialization:] This is the same as in the input scheme.

\item[Broadcasting:] Given a set ${\cal R}$ (which may be
empty at the first broadcasting after initialization),
the center executes the following:
\begin{enumerate}
\item Choose a session encryption key $K$.
\item Find a partition of the users
in ${\cal N} \setminus {\cal R}$ into disjoint subsets $S_{i_1}, \ldots, S_{i_m}$.
Let $L_{i_1},\ldots, L_{i_m}$ be the keys associated with the above subsets.
\item Encrypt $K$ with keys $f_{L_{i_1}}(0),\ldots, f_{L_{i_m}}(0)$ and send the ciphertext
$$
\langle [i_1, \ldots, i_m, E_{f_{L_{i_1}}(0)}(K), \ldots, E_{f_{L_{i_m}}(0)}(K)], F_K(M)\rangle.
$$
\item Update $L_{i}$ to $f_{L_i}(1)$ for all $i$ if ${\cal R} \neq \emptyset$.
\end{enumerate}

\item[Decryption:] A receiver $u$, upon receiving a broadcast message
$\langle [i_1, \ldots, i_m, C_1, \ldots, C_m], C \rangle$, executes as follows.
\begin{enumerate}
\item Find $i_j$ such that $u \in S_{i_j}$ (in the case $u \in
{\cal R}$ the result is {\sc null}). \item Extract the
corresponding key $L_{i_j}$ from $I_u$. \item Decrypt $C_j$ using
key $f_{L_{i_j}}(0)$ to obtain $K$. \item Decrypt $C$ using key
$K$ to obtain the message $M$. \item Update $L_i$ to $f_{L_i}(1)$
for all the $i$ it holds if the broadcast is incurred by a
revocation event (i.e., $m > 1$).
\end{enumerate}
\end{description}
\end{minipage}
}
\end{center}
\caption{The compiler for stateless group communication schemes}
\label{fig:stateless:strong-security}
\end{figure*}

\subsection{Security Analysis of the Compiler}
\label{sec:stateless:security-analysis}

The key idea that the scheme resulting from the above compiler is
not subject to the attack presented in the introduction is 
the following: compromise of a user at time $t$ does not
allow the adversary to recover keys corresponding to time $t_1 <t$.
This is fulfilled by updating the keys using an appropriate
family of pseudorandom functions.

\begin{theorem}
\label{theorem:stateless:compiler-security}
Suppose the input scheme is {\tt secure},
and $\{f_k\}$ is a secure pseudorandom function family,
and the different keys belonging to $\{L_i\}_i \cup \{I_u\}_u$ are
computationally independent of each other.
Then, the above scheme is {\tt strongly-secure}
in the sense of Definition \ref{def:stateless:strong-security}.
\end{theorem}

\begin{proof}
Consider a mental game in which the system is initialized as in
the input scheme. However, with respect to each broadcast
operation, each incorrupt $L_i$ is substituted with a pair of
independently chosen random keys $\langle L^{(a,0)}_i, L^{(a,1)}_i
\rangle$ such that $L^{(a,0)}_i$ is used to encrypt the message-encryption
key $K$ (if selected), where $a=1,2,\ldots$. We claim that this
scheme is {\tt secure}. This is because the keys that are used to
encrypt the session key are freshly and independently chosen at
random, which means that it is essentially a ``short-lived"
version of the input scheme.
We also claim that this scheme is {\tt strongly-secure}. This is because the keys that
are used to encrypt the message-encryption
key are freshly and independently chosen at random, which means that
the secrets compromised after the ``challenge" message are information-theoretically
independent of the the secrets used to encrypt the session key in the ``challenge" message.
Therefore, this scheme is {\tt strongly-secure}.

Suppose the scheme output by the compiler (called the real-life scheme) is not {\tt strongly-secure}.
We observe that the difference between the above mental game and the real-life scheme
is ``how the incorrupt keys are evolved." Specifically, for $a>0$, in the former case,
the $\langle L^{(a,0)}_i, L^{(a,1)}_i \rangle$ are independently chosen at random;
in the latter case,
$\langle L^{(a,0)}_i=f_{f^{a-1}_{L_i}(1)}(0), L^{(a,1)}_i = f^a_{L_i}(1)\rangle$,
where $f^0_X(\cdot)=X$, $f^1_X(\cdot)=f_X(\cdot)$, and $f^2_X(\cdot) = f_{f_{X}(\cdot)}(\cdot)$.

Now we consider the following experiment ${\sf EXPT}_j$, where $0 \leq j \leq \ell$ and $\ell$
is the total number of revocation operations (which is polynomially bounded).
The experiment is initialized as in the above mental game or as in the real-life scheme
(both are the same at this stage). For any $0 \leq a \leq j$, any {\em incorrupt}
$\langle L^{(a,0)}_i, L^{(a,1)}_i \rangle$ are independently chosen at random.
For any $j < a \leq \ell$, any {\em incorrupt}
$\langle L^{(a,0)}_i, L^{(a,1)}_i \rangle$ is defined as
$\langle L^{(a,0)}_i=f_{f^{a-j-1}_{L^{(j,1)}_i}(1)}(0),
L^{(a,1)}_i = f^{a-j}_{L^{(j,1)}_i}(1)\rangle$.
The experiments can get through because the secrets (some of them are used for encrypting
the message-encryption key) are (at least) computationally independent of each other.
We observe that ${\sf EXPT}_0$ corresponds to the real-life scheme,
and ${\sf EXPT}_\ell$ corresponds to the above mental scheme.
Since we assumed that ${\sf EXPT}_0$ is not {\tt strongly-secure}, it holds that
$\A$ has a non-negligible success probability
$\varepsilon_0$ with respect to Definition \ref{def:stateless:strong-security}.
On the other hand, we already know that ${\sf EXPT}_\ell$ is {\tt strongly-secure},
which means that $\A$ has only a negligible success probability $\varepsilon_\ell$.
Since $\ell$ is polynomially bounded, there must
exist $0 \leq j < \ell$ such that ${\sf EXPT}_j$ and ${\sf EXPT}_{j+1}$ are
distinguishable with a non-negligible probability (by the means of the adversary $\A$
that may or may not break the {\tt strong-security} in the respective experiments).
Suppose ${\sf f}$ is a challenge oracle that is either a random function or a pseudorandom function
with equal probability. Then, we can distinguish a random function from a pseudorandom one, via
black-box access to ${\sf f}$, with a non-negligible probability by
letting $\langle L^{(j,0)}_i, L^{(j,1)}_i \rangle$ be obtained from an
oracle query to ${\sf f}$
with respect to $L^{(j-1,1)}_i$. \qed
\end{proof}

\subsection{A Concrete {\tt strongly-secure} Stateless Group Communication Scheme}
\label{sec:stateless:concrete-scheme}

Within the subset-cover revocation framework, \cite{NaorCrypto01}
presented two concrete algorithms,
namely the {\em complete subtree method} and the {\em subset difference method}.
The difference between the two methods is how the
collection of subsets (covering ${\cal N} \setminus {\cal R}$)
is selected. Now we briefly review the {\em complete subtree method}, to which the above
compiler is applicable.

Suppose the receivers are the leaves in a rooted full binary tree with $N$ leaves
(assume that $N$ is a power of 2). Such a tree contains $2N-1$ nodes (leaves plus internal nodes)
and for any $1 \leq i \leq 2N-1$ we assume that $v_i$ is a node in the tree.
Denote by $ST({\cal R})$ the unique (directed) Steiner Tree
induced by the set ${\cal R}$ or vertices and
the root; i.e., the minimal subtree of the full binary tree that connects all the leaves in
${\cal R}$. The collection of subsets $S_1,\ldots,S_w$ in this scheme corresponds to all
complete subtrees in the full binary tree. For any node $v_i$ in
the full binary tree (either an internal node
or a leaf, $2N-1$ altogether) let subset $S_i$ be the collection of receivers $u$ that
correspond to the leaves of the subtree rooted at node $v_i$. In other words, $u \in S_i$ iff
$v_i$ is an ancestor of $u$.

The {\bf initialization} algorithm is simple: assign an {\em
independent} and random key $L_i$ to every node $v_i$ in the
complete tree, and provide every receiver $u$ with the $\log N+1$
keys associated with the nodes along the path from the root to
leaf $u$. (As said before, if the secret information $I_u$ is transmitted using a
key established via a two-party authenticated key-exchange protocol, then
the key is securely erased after the initialization.)
The {\bf broadcasting} algorithm is as follows. For a
given set ${\cal R}$ of revoked receivers, let $u_1,\ldots,u_r$ be
the leaves corresponding to the elements in ${\cal R}$. The method
to partition ${\cal N}\setminus {\cal R}$ into disjoint subsets is
as follows. Let $S_{i_1},\ldots,S_{i_m}$ be all the subtrees of
the original tree that ``hang" off $ST({\cal R})$; i.e., all
subtrees whose roots $v_1,\ldots,v_m$ are adjacent to nodes of
outdegree 1 in $ST({\cal R})$, but are not in $ST({\cal R})$.
It follows immediately that this collection covers all nodes in
${\cal N}\setminus {\cal R}$ and only those. As a result, in the
{\bf decryption} algorithm, given a message
$$
\langle [i_1,\ldots,i_m,E_{L_{i_1}}(K),\ldots,E_{L_{i_m}}(K)],F_K(M)\rangle
$$
a receiver $u$ needs to find whether any of its ancestors is
among $i_1,\ldots,i_m$; note that there
can be only one such ancestor, so $u$ may belong to at most one subset.

Note that the number of
subsets in a cover with $N$ users and $r$ revocations is at most $r \log \frac{N}{r}$.
The message length is of at most $r \log \frac{N}{r}$ keys.
Each receiver stores $\log N$ keys, and the center stores $2N-1$ keys.
The decryption process incurs $O(\log \log N)$ comparison
operations (for finding the cover) plus two decryption operations.

Proof of the following theorem can be straightforwardly adapted from \cite{NaorCrypto01}.

\begin{theorem}
\label{theorem:stateless:security}
The complete subtree revocation scheme of the subset-cover framework is {\tt secure}.
\end{theorem}

As showed before, the subset-cover framework, and thus the
complete subtree revocation scheme, is not {\tt strongly secure}.
As a corollary of Theorem
\ref{theorem:stateless:compiler-security} (which states that the
compiler in Section \ref{sec:stateless:compiler} can transforms a
{\tt secure} stateless group communication scheme into a {\tt
strongly-secure} one) and Theorem \ref{theorem:stateless:security}
(which states that the above complete subtree method is a {\tt
secure} stateless group communication scheme), the scheme output
by the compiler is {\tt strongly-secure}.

\begin{corollary}
The stateless group communication scheme obtained by applying the
compiler in Section \ref{sec:stateless:compiler} to the above {\tt
secure} complete subtree revocation scheme is {\tt
strongly-secure}.
\end{corollary}

Now we analyze the {\em extra} complexities (corresponding to
each revocation event) for achieving
{\tt strong-security}.
\begin{itemize}
\item The center updates its keys by evaluating $2N-1$ pseudorandom functions
(this corresponds to the worst case scenario that no keys have
been corrupt -- the corrupt keys, if known,
do not need to be updated). Moreover, in order to encrypt a message, the center needs to evaluate
$r \log \frac{N}{r}$ pseudorandom functions; this computational complexity can indeed be traded
with an extra $2N-1$ storage complexity. Since the center is typically powerful in terms of
computation, communication, and storage, these extra complexities are insignificant.
\item Each receiver needs to evaluate $2\log N$ pseudorandom functions and at most stores
$\log N$ keys. Even if the receivers are low-end equipment (e.g., sensors), these extra
complexities should still be insignificant.
\end{itemize}

\subsection{Discussions}
\label{sec:stateless:discussion}

The class of the stateless group communication schemes that can be
made {\tt strongly-secure} via the compiler in Section
\ref{sec:stateless:compiler} should possess the following
property: all the different keys belonging to $\{L_i\}_i \cup
\{I_u\}_u$ are computationally-independent of each
other. This explains why the above compiler applies to the {\em
complete subtree method} of \cite{NaorCrypto01}. On the other
hand, the {\em subset difference method} of \cite{NaorCrypto01},
which does not achieve the desired {\tt strong-security}, cannot
made {\tt strongly-secure} via the above compiler because the keys
belonging to $\{L_i\}_i \cup \{I_u\}_u$ are not computationally
independent.\footnote{The independence condition can indeed be
satisfied at the expense of each receiver storing $O(N)$ keys,
which is clearly not scalable.}

The stateless group communication schemes presented in
\cite{JhoEurocrypt05}, which outperforms
\cite{NaorCrypto01,HalevyCrypto02} under certain interesting
circumstances, are not {\tt strongly-secure}. Unfortunately, they
cannot be made {\tt strongly-secure} via the above compiler for a
similar reason. It is an interesting open question to make the
stateless group communication schemes of \cite{JhoEurocrypt05}
{\tt strongly-secure} at an expense similar to the extra
complexity imposed by the compilers presented in this paper.

\section{Conclusion and Open Problems}
\label{sec:conclusion}

We showed that a class of existing group communication schemes,
stateful and stateless alike, are vulnerable to a realistic severe attack.
We presented formal models that allow us to capture the desired security properties,
and explore the relationships between the security notions.
We showed how some methods can make a {\em subclass} of existing schemes immune to the attack
at a very small extra cost. An interesting open question is to
make other schemes (e.g., the stateful \cite{ShermanOFT03,BalensonOFTDraft}
and the stateless \cite{JhoEurocrypt05}) secure against the attack without imposing any
significant extra complexity.

\section*{Acknowledgements}
We thank Jonathan Katz for illuminating discussions that led to the refined model
in Section \ref{sec:model}, our SASN'05 shepherd, Donggang Liu, for helpful feedback
and communication, and the SASN'05 anonymous reviewers for useful comments.
We thank the anonymous reviewers of this special issue for detailed and
constructive suggestions that improved this paper. We thank Paul Parker for
a careful proofreading that helped polish the writing.

This work was supported in part by ARO, NSF, and UTSA.


\begin{appendix}

\section{Join and Leave Protocols of LKH}
\label{appendix:wgl}

For completeness, we briefly review join and leave protocols of LKH
in Figure \ref{fig:wgl}. The notations are consistent with
the main body of the paper.

\begin{figure*}[ht]
\begin{center}
\fbox{
\begin{minipage}{41pc} 
\begin{tabbing}
123\=123\=123\=\kill
Join protocol for group-oriented rekeying: \ \ \ \ \ \ \ \ // suppose user $u$ joins the group \\
\> server $s$ generates a new key $k_u$ for user $u$ \\
\> server $s$ finds a joining point $x_{j}$ \\
\> server $s$ attaches $k_u$ to $x_j$ \\
\> let $x_0$ be the root \\
\> $\bk_{j+1} \leftarrow k_u$ \\
\> denote by $x_{i-1}$ the parent of $x_i$ for $1 \leq i \leq j$ \\
\> let $\bk_0,\bk_1, \ldots,\bk_j$ be the current keys of $x_0,\ldots,x_j$, respectively \\
\> server $s$ generates fresh keys $\hk_0,\hk_1, \ldots, \hk_j$ \ \ \ \ //
new keys of $x_0,\ldots,x_j$ \\
\> $s \to \userset(\bk_0): \{\hk_0\}_{\bk_0}, \{\hk_1\}_{\bk_1}, \ldots,\{\hk_j\}_{\bk_j}$ \\
\> $s \to u: \{\hk_0, \hk_1, \ldots, \hk_j\}_{k_u}$
\end{tabbing}
\begin{tabbing}
123\=123\=123\=\kill
Leave protocol for group-oriented rekeying: \ \ \ \ \ \ \ \ // suppose $u$ leaves the group \\
\> let $x_{j+1}$ be the deleted $k$-node for $k_u$ \\
\> $\bk_{j+1} \leftarrow k_u$ \\
\> server $s$ finds the leaving point $x_j$ (parent of $k_u$)  \\
\> server $s$ removes $\bk_{j+1}$ from the key tree \\
\> let $x_0$ be the root \\
\> denote by $x_{i-1}$ the parent of $x_i$ where $1 \leq i \leq j$ \\
\> let $\bk_0,\bk_1, \ldots, \bk_j$ be the keys of $x_0, x_1, \ldots,x_j$
\ \ \ // they need to be changed \\
\> server $s$ generates fresh keys $\hk_0,\hk_1, \ldots, \hk_j$ as
the new keys of $x_0, x_1, \ldots,x_j$ \\
\> {\tt FOR} $i = 0 \text{ \tt TO } j$ \\
\>\> let $\bk_{i_1},\ldots \bk_{i_{z_i}}$ be the keys
at the children of $x_i$ in the new key tree \\
\>\> $L_i \leftarrow (\{\hk_i\}_{\bk_{i_1}}, \ldots, \{\hk_i\}_{\bk_{i_{z_i}}})$ \\
\> $s \to \userset(\bk_0) \setminus \{u\}: (L_0, \ldots, L_j)$
\end{tabbing}
\end{minipage}
}
\end{center}
\caption{Join- and leave-incurred group-oriented rekeying in LKH}
\label{fig:wgl}
\end{figure*}

\end{appendix}

\end{document}